\documentclass[twoside,11pt]{article}

\usepackage{jmlr2e}

\usepackage{microtype}
\usepackage{graphicx}
\usepackage{subfigure}
\usepackage{booktabs} 
\usepackage[T1]{fontenc}
\usepackage[a4paper,top=3cm,bottom=2cm,left=3cm,right=3cm,marginparwidth=1.75cm]{geometry}
\usepackage{hyperref}

\usepackage{float}
\usepackage{wrapfig}
\usepackage{algorithm}
\usepackage[noend]{algorithmic}
\usepackage{amsmath}
\usepackage{amssymb}
\usepackage{bm}
\usepackage{color}
\usepackage{enumerate}
\usepackage{graphicx}
\usepackage{subfigure}
\usepackage{natbib}
\usepackage{mathtools}
\usepackage[usenames,dvipsnames]{xcolor}
\usepackage{booktabs}       
\usepackage{amsfonts}       
\usepackage{nicefrac}       

\newcommand{\PP}{\mathbb{P}}

\newcommand{\Prob}[1]{\mathbb{P}\left(#1\right)}
\newcommand{\E}{\mathbb{E}}
\newcommand{\one}[1]{\mathbb{I}\left\{#1\right\}}
\DeclareMathOperator{\sgn}{sgn} 
\DeclareMathOperator{\argmin}{argmin}
\DeclareMathOperator{\argmax}{argmax}

\newcommand{\R}{\mathbb{R}}
\newcommand{\N}{\mathbb{N}}
\newcommand{\pms}{\{\pm 1\}}
\newcommand{\oR}{\overline{R}}

\newcommand{\Loss}{\mathcal{L}}
\newcommand{\defeq}{\doteq}

\newcommand{\norm}[1]{\left\|#1\right\|}

\newcommand{\Fnorm}[1]{\norm{#1}_{\mathrm{F}}}

\def\ao#1{{{\color{black}#1}}}

\def\aor#1{{{\color{black}#1}}}
 \usepackage[capitalize,noabbrev]{cleveref}
\usepackage{algorithm}
\usepackage{algorithmic}

\usepackage{amsfonts}

\usepackage{lastpage}
\ShortHeadings{Gradient Descent for Rank-One Matrix Completion}{Ma, Olshevsky, Saligrama, Szepesvari}
\firstpageno{1}

\begin{document}

\title{Gradient Descent for Sparse Rank-One Matrix Completion for Crowd-Sourced Aggregation of Sparsely Interacting Workers\thanks{This work was supported partly by the National Science Foundation Grant 1527618, the Office of Naval Research Grant N0014-18-1-2257, the Hariri Institute's Data Science Fellowship and by a gift from the ARM corporation.}}

\author{
\name Yao Ma \email yaoma@bu.edu \\ 
\addr Division of Systems Engineering \\ 
Boston University \\ 
\AND
\name Alex Olshevsky \email alexols@bu.edu \\
       \addr Department Electrical and Computer Engineering\\
             Division of Systems Engineering \\ 
       Boston University\\
       \AND
       \name Venkatesh Saligrama \email srv@bu.edu \\ 
       \addr Department Electrical and Computer Engineering\\
       Boston University\\
   	   \AND
       \name Csaba Szepesvari \email szepi@google.com \\
       \addr Google Deepmind}
   
\editor{Inderjit Dhillon}
\maketitle

\begin{abstract}%
We consider worker skill estimation for the single-coin Dawid-Skene crowdsourcing model. 
In practice, skill-estimation is challenging because worker assignments are sparse and irregular due to the arbitrary and uncontrolled availability of workers. We formulate skill estimation as a rank-one correlation-matrix completion problem, where the observed components correspond to \emph{observed} label correlation between workers. We show that the correlation matrix can be successfully recovered and skills are identifiable if and only if the sampling matrix (observed components) does not have a bipartite connected component. We then propose a projected gradient descent scheme and show that skill estimates converge to the desired global optima for such sampling matrices. Our proof is original and the results are surprising in light of the fact that even the weighted rank-one matrix factorization problem is NP-hard in general. Next, we derive sample complexity bounds in terms of spectral properties of the {\it signless} Laplacian of the sampling matrix. Our proposed scheme achieves state-of-art performance on a number of real-world datasets.

\end{abstract}

\begin{keywords}
  distributed optimization, stochastic gradient descent
\end{keywords}

\section{Introduction}
Crowdsourcing can be a scalable approach to collecting data for tasks that require human knowledge such as image recognition and natural language processing. Through crowdsourcing platforms such as  Amazon Mechanical Turk, a large number of data tasks can be assigned to workers who are asked to give binary or multi-class labels. The goal of much of crowdsourcing research is to estimate the unknown ground truth, given that the quality of the workers can be variable. Indeed, due to the high variability of worker skills, aggregating true labels becomes a challenging problem.

One straightforward approach is to directly estimate the unknown labels by majority voting from the information provided by workers. In this approach, an implicit assumption is that all workers have identical skills on each task; on the other hand, one might expect the answers from reliable workers are more likely to be accurate. In practice, the crowd is often highly heterogeneous in terms of skill levels, and downweighting unskilled workers and upweighting skilled workers can have  a significant impact on the performance. Many aggregation methods ranging from the weighted majority vote to more complex schemes that incorporate worker quality and accuracy have been proposed. Theoretically, recent works \citep{BeKo14:NIPS,Sze15:MSc} have investigated the importance of having precise knowledge of skill quality for accurate prediction of ground-truth labels. Moreover,  accurate skill estimation can also be useful for other purposes like worker training, task assignment, or for use in worker-compensation schemes. 

There are two challenges in estimating skills of workers given that the problem setup is unsupervised. The first challenge is to construct a  skill model for each worker. Many papers achieve empirical success by applying Dawid \& Skene (DS) model \citep{DaSke79}, which is a simple model that parameterized by the probability of a worker answers the true label. In this paper, the basis of our works is the homogeneous DS model where each worker is assumed to have the same skill level on each class. More specifically, we focus on the single-coin (DS) model for binary crowdsourcing problem in this paper (though in Section~\ref{sec:WLS}, we extend our algorithm to multiclass problems).

The second challenge is that, in practice, workers are only available for a short period of time which means only a small subset of data is labeled by each worker. This introduces a  sparse worker-task assignment \citep{Karger2013,dalvi_aggregating_2013}. An additional subtle issue is the lack of diversity in terms of interactions between the workers: a worker is often grouped with a limited subset of workers across all tasks. This situation is remarkably evident on benchmark datasets: The 'Web' dataset has $177$ workers, with $3$ to $20$ workers/task and each worker on average interacting with about another $2.7$ workers only, while the standard deviation of how many workers a worker is interacting with is $15$. The 'RTE' dataset has $164$ workers, has only $10$ workers/task on the average and each worker interacts with fewer than $2.5$ other workers, while the standard deviation of the interaction degree is $20$. This is in contrast to most existing crowdsourcing research which only considered estimate skills with nearly complete data. We are therefore motivated by the need to make spectral methods suitable for non-regular worker-task data often seen in practice. 

In this paper, we suppose that the input comes in the form of a \emph{sparsely filled} 
$W\times T$ worker-task label matrix.
The workers possess unique unknown skills, and tasks assume unique unknown labels. The worker-task label matrix collects the random labels provided by the workers for individual tasks. The skill level of a worker is the (scaled) probability of the worker's label matching the true unknown label for any of the tasks.
The observed labels are independent of each other. 

Given the workers' skill levels, the optimal way \citep{NiPa81:WM,shapley_optimizing_1984} to reconstruct the unknown labels is to use weighted majority voting where the weights assigned to the label provided by a worker is equal to the log-odds underlying the worker's skill. Since skill levels are unknown, we follow prior works
\citep{dalvi_aggregating_2013,BeKo14:NIPS,Sze15:MSc,BoCo16} and adopt a two-step approach, whereby worker skills are first estimated and then these skills are used with the optimal weighting method to recover labels. Our main contributions are as follows:

\begin{enumerate}
    \item We construct a skill estimator under single-coin model as a weighted least-squares rank-one matrix completion/factorization problem. The matrix being factored is the correlation matrix among the workers, with the weights compensating for the varying accuracy in the inter-worker correlations. 
    
    \item We show that skills can be recovered from the observation data matrix whenever the worker-worker interaction graph does not contain a bipartite connected component. In particular, for any crowdsourcing problem that has non-bipartite worker-worker interaction graph, there always exists a method to estimate true skills. 
    
    \item In the context of minimizing the objective function, we propose to use projected gradient descent which is theoretically verified to converge to the true skills. We give natural and mild conditions on the weighting matrix under which we prove that projected gradient descent,  despite the objective being non-convex, is guaranteed to find the rank-one decomposition of the true moment matrix. 
    
    \item We extend our algorithm to multiclass case by applying the homogeneous DS model. Under this model, we prove that any multiclass problem can be formulated as a weighted least-squares rank-one problem  where the unknown variable is a linear function of true skills.

\end{enumerate}

{\em Our approach is also of independent interest, as we derive a fundamental result about  symmetric rank-one matrix completion: the unobserved entries can be recovered by gradient descent in polynomial time whenever the sampling matrix is irreducible and non-bipartite}.  Our results for convergence of the proposed gradient descent scheme should be somewhat surprising given that the related weighted low-rank factorization problem is known to be NP-hard even for the rank-one case \citep{GiGli11:NPhardness}. 
In contrast to our approach, existing results in low-rank matrix completion require strong assumptions on the weighting matrix, typically some form of incoherence, e.g., \citep{Rong}.

\section{Related Work}

{\it Discriminative Approach:} In contrast to our two-step approach, several works adopt a discriminative method for label prediction. Specifically, \citet{Li2011,tian2015} directly identify true labels by various aggregation rules that incorporate worker reliability. 

{\it Skill Estimation:} As mentioned earlier, we work in the problem of estimating skills under the single-coin model. Past approaches to skill estimation are based on \emph{maximum likelihood/maximum posteriori} (ML/MAP) estimation, or \emph{moment matching}, or a combination of these. 
In particular, various versions of the EM algorithm have been proposed
to implement ML/MAP estimation, starting with the work of
\citet{DaSke79}. 
Variants and extensions of this method, tested in various problems, include
\citet{hui1980estimating,smyth1995inferring,albert2004cautionary,raykar2010learning,liu2012variational}.
A number of recent works were concerned with performance guarantees for Expectation Maximization (EM) and some of its variants
\citep{gao2013minimax,Zhang2014,GaoLZ16}.
Another popular direction is to add priors over worker skills, labels or worker-task assignments.
To properly deal with the extra information, various Bayesian methods (belief propagation, mean-field and variational methods) have been considered 
\citep{raykar2010learning,karger_iterative_2011,liu2012variational,Karger2013,karger_budget-optimal_2014}. 
Moment matching is also widely used 
\citep{ghosh2011moderates,dalvi_aggregating_2013,Zhang2014,GaoLZ16,BoCo16,zhang_spectral_2016}. 
With the exception of \citet{BoCo16}, who propose an ad-hoc method,
the algorithms in these works use matrix or tensor factorization.%
\footnote{While \citet{ghosh2011moderates} pioneered the matrix factorization approach, their work is less relevant to this discussion as they estimate the labels directly.}

In theory, an ML/MAP method which is \emph{guaranteed} to maximize the likelihood/posterior, is the ideal method to accommodate irregular worker-task assignments.
However, as far as we know, none of the existing algorithms, unless initialized with a moment-matching-based spectral method, is proven to indeed find a satisfactory approximate maximizer of the objective that it is maximizing \citep{zhang_spectral_2016}.
At the same time, moment matching methods that use spectral (and in general algebraic) algorithms implicitly assume the regularity of worker-task assignments, too.
Indeed, the approach of \citet{ghosh2011moderates} crucially relies on the regularity of the worker-task assignment (as the method proposed uses unnormalized statistics). In particular, this method is not expected to work at all on non-regular data.
Other spectral methods, being purely algebraic, implicitly treat all entries in the estimated matrices and tensors as if they had the same accuracy, which,
in the case of irregular worker-task assignments, is far from the truth. In particular, the need to explicitly deal with data with unequal accuracy is a widely recognized issue that has a long history in the low-rank factorization community, going back to the work of \citet{gabriel_lower_1979}.
Starting with this work, 
the standard recommendation is to reformulate the low-rank estimation problem as a weighted least-squares problem \citep{gabriel_lower_1979,srebro_weighted_2003}.
In this paper, we will also follow this recommendation.

While \citet{dalvi_aggregating_2013} also use a weighted least-squares objective, this is not by choice, but rather as a consequence of the need to normalize the data rather than to correct for the inaccuracy of the data. Furthermore, rather than considering the direct minimization of the resulting objective, they use two heuristic approaches that also use an unweighted spectral method.

In this light, our goal is to make spectral methods suitable for non-regular worker-task data often seen in practice. 

{\it Matrix Factorization/Completion:} Unlike the general matrix factorization problem arising in recommender systems \citep{koren}, we are primarily concerned with a rank-one estimation of square symmetric matrices. Existing results on matrix completion \citep{Rong} for square symmetric matrices are more general but require stronger assumptions on the matrix such as incoherence and random sampling.

\bigskip

\textbf{Notation and conventions}: 
The set of reals is denoted by $\R$, the set of natural numbers which does not include zero is denoted by $\N$. 
For $k\in \N$, $[k] \defeq \{1, \dots, k \}$.
Empty sums are defined as zero. We will use $\PP$ to denote the probability measure over the measure space
holding our random variables, while $\mathbb{E}$ will be used to denote the corresponding expectation operator.
For $p\in [1,\infty]$, we use $\norm{v}_p$ to denote the $p$-norm of vectors. Further, $\norm{\cdot}$ stands for the $2$-norm, $\Fnorm{\cdot}$ is the Frobenius-norm.
The cardinality of a set $S$ is denoted by $|S|$. 
For a real-valued vector $x$, $|x|$ denotes the vector whose $i$th component is $|x_i|$.
Proofs of new results, missing from the main text are given in the appendix.

\vspace{-0.1in}
\section{Formal problem statement\label{sec:prob}}
\vspace{-0.05in}
We first consider binary crowdsourcing tasks where a set of workers provide $\{-1,1\}$ labels for a large number of items. Let $W\in \N$ be a positive integer denoting the number of workers.  A problem instance $\theta \doteq (s,A,g)$ 
is given by: a skill vector 
$s = (s_1,\dots,s_W)\in [-1,1]^W$ associating the skill level $s_w$ with worker $w$; 
the worker-task assignment set $A\subset [W]\times \N$, which captures which workers provide labels on which tasks; and the vector of ``ground truth labels'' $g\in\pms^{\N}$, which are unknown and which we would like to estimate.

When $A\subset [W]\times [T]$ for some $T\in \N$, we say that $\theta$ is a finite instance with $T$ tasks; 
otherwise we will say $\theta$ is an infinite instance. 
We allow infinite tasks to be able to discuss asymptotic identifiability.

It will be convenient to use $\Theta_W$ to denote the set of all possible problem instances, defined as above. For any instance $\theta \in \Theta_W$, the worker-task assignment set provides important information about worker interaction structure. Indeed, we can think of two workers as ``interacting'' if they provide a label for the same task. Formally, we define the interaction graph as follows. 

\begin{definition}[Interaction Graph] \label{intergraph}
Let $A$ be a worker-task assignment set. 
The (worker) interaction graph underlying $A$
is an undirected graph $G = G_A$ with vertex set $[W]$ 
such that $G= ([W],E)$ with $(i,j)\in E$ 
if there exists some task $t\in \N$ such that both $(i,t)$ and $(j,t)$ are elements of $A$.
\end{definition}
In the case of infinite instance, the interaction graph is a unweighted graph where an edge does not have any weight associated with it. For finite instances, it will make sense to assign a weight on the edge $(i,j)$ which is  the number of tasks shared by workers $i$ and $j$.

Our goal 
is to recover the ground truth labels $(g_t)_t$ 
given observations $(Y_{w,t})_{(w,t)\in A}$, where $Y_{w,t}$ is a $\pm 1$ random variable associated with worker $w$ and task $t$. According to the single-coin model, 
the observations are generated as
$Y_{w,t} = Z_{w,t} g_t$, 
where $(Z_{w,t})_{(w,t)\in A}$ is a collection of 
mutually independent random variables that satisfy $\E[Z_{w,t}]=s_{w}$. Note that this is the same as assuming that worker $w$ returns $g_t$ with probability $(1+s_w)/2$ and $-g_t$ with probability $(1-s_w)/2$.

Thus a worker $w$ with $s_w=1$ always returns the ground truth, while a worker with $s_w=-1$ always return the opposite label of the ground truth; and a worker with $s_w=0$ will always provide a random variable $Y_{w,t}$ with zero expectation regardless of the ground truth, i.e., a uniformly random label. 

\bigskip

\noindent {\bf Remark:} As will be discussed in detail later, some additional assumptions on the skill vector $s$ will be needed for accurate estimation of ground truth labels; obviously, little can be done if $s$ is the zero vector, i.e., if every worker returns uniformly random labels irrespective of the ground truth. Another obvious observation is, in the event that all workers agree, we cannot distinguish the possibility that $s$ is proportional to the all ones vector (and every worker provides the right label) from the possibility that $s$ is proportional to the negative of the all ones vector (and every worker provides the wrong label). One way to get around this problem is to assume that $s>0$, i.e., all workers have at least some skill; we will make this assumption when analyzing our projected gradient descent method. A weaker approach is to assume that $\sum_{w} s_w > 0$, i.e., that, on the net, the workers are collectively more prone to return correct rather than incorrect labels; as we discuss later, this is sufficient for identifiability. 

\bigskip


A (deterministic) \emph{inference method} underlying an assignment set $A$
takes the observations 
$(Y_{w,t})_{(w,t)\in A}$ and 
returns a real-valued score for each task in $A$;
the signs of the scores give the label-estimates.
Inference methods are aimed at working with finite assignment sets.
To process an infinite assignment set, we define the notion of \emph{inference schema}.
In particular, an \emph{inference schema} underlying an infinite assignment set $A$ is defined
as the infinite sequence of inference methods $\gamma^{(1)},\gamma^{(2)},\dots$ such that 
$\gamma^{(t)}$ is an inference method for the first $t$ tasks.

When important, we will use the subindex $\theta$ in $\PP_\theta$ to denote the probability when the problem instance is $\theta$.  We will use $\E_\theta$ to denote
the corresponding expectation operator.
With this notation,
the expected loss suffered by an inference schema $\gamma=(\gamma^{(1)},\gamma^{(2)},\dots)$ 
on the first $T$ tasks of an instance $\theta$ is
\[
\Loss_T(\gamma;\theta) = 
	\frac1T\, \E_{\theta}\Bigl[ 
     \textstyle	\sum_{t=1}^T\one{\gamma^{(T)}_t(Y) g_t \leq 0} \Bigr]\,.
\]
The optimal inference schema for an assignment set $A$ \emph{given} the knowledge 
of the skill vector $s\in [-1,1]^W$ is denoted by $\gamma^*_{s,A}$. The next section gives a
simple explicit form for this optimal schema.
The \emph{average regret} 
of an inference schema $\gamma = (\gamma^{(1)},\gamma^{(2)},\dots)$ for an instance $\theta\in \Theta$ is its excess loss on the instance $\theta$ as compared to the loss of the optimal schema:
\[
\oR_T(\gamma;\theta) = \Loss_T(\gamma;\theta) - \Loss_T(\gamma^*_{s,A};\theta)\,.
\]
If the average regret converges to zero, then the loss suffered by $\gamma$ asymptotically converges to the loss of the optimal inference.  Based on this, we define asymptotic consistency and learnability:
\begin{definition}[Consistency and Learnability]
An inference schema is said to be (asymptotically) \emph{consistent} for an instance set $\Theta\subset \Theta_W$
if, for any $\theta\in \Theta$,  $\limsup_{T\to\infty} \oR_T(\gamma)=0$.
An instance set $\Theta\subset \Theta_W$ is (asymptotically) \emph{learnable} if there is a consistent 
inference schema for it.
\end{definition}

\subsection{Two-Step Plug-in Approach}
\label{Sec:twostep}
In this work we will pursue a two-step approach based on first estimating the skill vector $s$ and then utilizing a plug-in classifier to predict the ground-truth labels. The motivation for a two-step approach stems from existing results that characterize accuracy in terms of skill estimation errors. 
For the sake of exposition, we recall some of these results now.

It has been shown in \citet{Li2011}, the optimal classifier is log-odds weighted majority voting given by the MAP rule. Suppose the prior distribution of true labels is the uniform distribution over $\{-1,1\}$; then the Bayes classifier is well-known to be the optimal classifier \citep{Duda2012},i.e.,
\begin{align}
\gamma_{s,A}^*(Y_t)=&\argmax_{l\in\{+1,-1\}}\PP(g_t=l|Y_t,s,A) \nonumber \\
=&\argmax_{l\in\{+1,-1\}}\PP(g_t=l)\PP(Y_t|g_t,s,A) \nonumber \\
=&\argmax_{l\in\{+1,-1\}}\log\PP(Y_t|g_t,s,A)  \nonumber \\ 
=& \argmax_{l\in\{+1,-1\}}\sum_{i=1}^W\log{\frac{1+s_i}{1-s_i}}\one{Y_{i,t}=l}, \label{eq:plugin}
\end{align}
where the third equation follows the assumption that $\PP(g_t=+1)=\PP(g_t=-1)=1/2$, and the fourth equation, after some algebra, is compact write to write the MAP estimator. Notice that $\gamma^*_{s,A}$ is a function of only one parameter, namely the skill vector $s$.  

Regarding the loss of the optimal schema, we start with introducing result when skills are known in advance. In this case, \citet{BeKo14:NIPS} provides an upper error bound, as well as an asymptotically matching lower  bound, which are stated as follows:
\begin{lemma}
For any task $t\in\N$, the optimal decision rule $\gamma^*$ satisfies
\begin{align*}
\Prob{\gamma^*_t(Y)\neq g_t}\leq&\exp{\left(-\frac{1}{2}\Phi\right)},\\
\Prob{\gamma^*_t(Y)\neq g_t}\geq&\frac{3}{4[1+\exp{(2\Phi+4\sqrt{\Phi})}]},
\end{align*}
where $\Phi=\sum_{i\in W}s_{i}\log{\left(\frac{1+s_i}{1-s_i}\right)}$ is called \emph{committee potential}.
\end{lemma}

However, we do not assume that we know the skills of workers in reality, and thus the true optimal inference
classifier is unknown to us. One natural way is to construct a true label inference that approximates the optimal Bayes classifier via estimating workers' skills. Fortunately, in addition to the case of known skills, 
 \citet{Sze15:MSc,BeKo14:NIPS} also provide an error bound when skills are only estimated:
\begin{lemma}
\label{Lem:predictionerror}
For any $\epsilon > 0$, the loss with estimated weights $\hat{v}_i = v(\hat s_i)$ satisfies 
\begin{align*}
&\frac1T\, \E_{\theta}\Bigl[ 
    	\sum_{t=1}^T\one{\gamma_{t,\hat s}(Y) g_t \leq 0} \Bigr]\\ 
        &\leq \frac1T\, \E_{\theta}\Bigl[ 
    	\sum_{t=1}^T\one{\gamma^*(Y) g_t \leq \epsilon} \Bigr] + \PP_{\theta}(\|v^*-\hat{v}\|_1 \geq \epsilon)\,.
\vspace{-0.15in}
\end{align*}
\end{lemma}
In turn, the error $\norm{v^*-\hat{v}}_1$ can be bounded in terms of the
multiplicative norm-differences in the skill estimates (see \citet{BeKo14:NIPS}):
\begin{lemma} \label{lem:errorV}
Suppose $\frac{1+ \hat s_i}{1+s_i}, \frac{1- \hat s_i}{1-s_i} \in [1-\delta_i, 1+\delta_i]$ then 
$|v(s_i) - v(\hat s_i)| \leq 2|\delta_i|$.%
\end{lemma}

These results together imply that a plug-in estimator with a guaranteed accuracy on the skill levels, in turn, leads to a bound on the error probability of predicting ground-truth labels. This motivates the skill estimation problem, which we consider in the remainder of this paper. 
\section{Weighted Least-Squares Estimation} \label{sec:WLS}
In this section, we propose an asymptotically consistent skill estimator for potentially sparse worker-task assignments. We are motivated by the scenario when, for most workers, only a very small portion of tasks are assigned to them. This induces  not only an extremely sparse worker-task assignment graph, but more importantly a sparse worker-worker interaction graph.

Recall that given a problem instance $\theta=(s,A,g)$, 
the data of the learner is given by the  matrix
$(Y_{i,t})_{(i,t)\in A}$ which is a collection of independent binary random variables such that $Y_{i,t}=g_tZ_{i,t}$ and $s_i=\mathbb{E}(Z_{i,t})$. When $A$ is finite, we define $N\in \N^{W\times W}$ to be the matrix whose $(i,j)$th entry with $i\ne j$ gives the number of times the workers $i$ and $j$ labeled the same task:
\begin{align*}
N_{ij}=|\{t\in \N:(i,t),(j,t)\in A\}|\,
\end{align*}
and we also let $N_{ii} = 0,\forall i=1,\ldots,W$. Note that there is an edge between workers $i$ and $j$ in the interaction graph exactly when $N_{ij}>0$. 

When $A$ is infinite, $N_{ij}$ may be infinite.
In this case, for $i\ne j$ we also define
$N_{ij}(T)=|\{t\in [T]:(i,t),(j,t)\in A\}|$ to denote the number of times workers $i$ and $j$ provide a label for the same task in the first $T$ tasks, and similarly we let  $N_{ii}(T)=0$ for all $i$.

The starting point of our approach is the following observation about the single coin model: the expected correlation between each pair of workers is ground-truth independent. Indeed,
\begin{eqnarray*} \mathbb{E}[Y_{i,t} Y_{j,t}]  & = &  \mathbb{E}[ g_t Z_{i,t} g_t Z_{j,t} ] \\ 
 & = &  \mathbb{E}[ Z_{i,t} Z_{j,t} ] \\ 
 & = &  \left( \frac{1+s_i}{2} \frac{1+s_j}{2} + \frac{1-s_i}{2} \frac{1-s_j}{2} \right) \cdot 1 \\
 &   &  + \left( \frac{1+s_i}{2} \frac{1-s_j}{2} + \frac{1-s_i}{2} \frac{1+s_j}{2} \right) \cdot (-1) \\ 
 & = &  s_i s_j,
\end{eqnarray*} where the second equation used that $g_t^2=1$. 

This observation motivates estimating the skills using 
{\small
\begin{equation}
\label{Eq:Objective}
\tilde{s}=\argmin_{x\in[-1,+1]^W}\frac{1}{2}\sum_{i, j, t ~|~ (i,t),(j,t)\in A}(Y_{i,t}Y_{j,t}-x_{i}x_{j})^2
\end{equation}}

Note that the number of terms containing the skill estimate $x_i$ of particular worker $i$ in this objective
scales with how many other workers this worker $i$ works with.
Intuitively, this should feel ``right'': the more a worker works with others, the more information we should have about its skill level. 

As it turns out, there is an alternative form for this objective, which is also very instrumental and which will form the basis of our algorithm and also of our analysis.
To introduce this form, define
$C_{ij} \defeq s_i s_j$ 
and let its empirical estimation be
\begin{equation} \label{samp_corr}
\tilde{C}_{ij}=\frac{1}{N_{ij}}\textstyle \sum_{t ~|~ (i,t),(j,t)\in A}\, Y_{i,t}Y_{j,t}\,.
\end{equation}
An alternative form of the objective in Eq.~\eqref{Eq:Objective} is given by the following result:
\begin{lemma}\label{lem:equiv}
Let $L:[-1,1]^W \to [0,\infty)$ be defined by 
\[
L(x) = \frac{1}{2}\sum_{(i,j)\in E}N_{ij}(\tilde{C}_{ij}-x_ix_j)^2.
\]
The optimization problem of Eq.~\eqref{Eq:Objective} is equivalent to the optimization problem
\[\argmin_{x\in[-1,+1]^W}L(x).\]
\end{lemma}
The proof, which is just simple algebra to show the two objective functions are equal up to a constant shift, is given in \cref{sec:equiv}.

The objective function from Lemma \ref{lem:equiv} can be seen as a weighted low-rank objective,
first proposed by \citet{gabriel_lower_1979}.
Clearly, the objective prescribes to approximate $\tilde{C}$ using $x x^\top$, with the
error in the $(i,j)$th entry scaled by $N_{ij}$. Note that this weighting is reasonable
as the variance of $\tilde{C}_{ij}$ is proportional to $1/N_{ij}$ and we expect from the theory of least-squares
that an objective combining multiple terms where the data is heteroscedastic (has unequal variance), the terms should be weighted with the inverse of the data variances.
Since $N_{ii}=0$, the weighting function $N$ can in general be full-rank, and in this case the general weighted rank-one optimization approximation known to be NP-hard \citep{GiGli11:NPhardness}. 

However, our data has special structure, 
which will allow one to avoid the existing hardness results. Indeed, on the one hand, as the number of data points increases, $\tilde{C}_{ij}$ will be near rank-one itself; and, on the other hand, we will put natural restrictions on the weighting matrix which are in fact necessary for identifiability. 
These conditions will allow us to avoid the NP-hardness results of \citet{GiGli11:NPhardness}.



\subsection{Plug-in  Gradient Descent}

To solve the weighted least-squares objective, the simplest algorithm is the gradient descent algorithm. 
We propose a \emph{Plug-in Gradient Descent} (PGD) algorithm that sequentially updates the skill level based on following the (negative) gradient of the loss $L$ at each time step:
\begin{align*}
\tilde{x}^{t+1}_i=&x^{t}_i+\eta\sum_{(i,j)\in E} N_{ij} (\tilde{C}_{ij}-x^{t}_ix^{t}_j)x^t_{j}
\end{align*}
%
where $N_i=|\{t:(i,t)\in A\}|=\sum_{j=1}^WN_{ij}$ is the number of tasks labeled by worker $i$ and $\tau>0$ is a tuning parameter. We do not necessarily need to explicitly enforce the constraint that $x \in [-1,1]^n$, though we have found that it helps in terms of the practical performance of the method, as we'll remark in Section \ref{sec:exp} where we discuss experimental results.

\subsection{An Extension to Multi-class Classification}

We now briefly describe how our approach may be extended to the case when the labels are not binary. Above, we have shown how the binary case may be reduced to a noisy rank-one matrix completion problem as in Lemma \ref{lem:equiv}. Here we show how the same approach can be used for the multiclass case. 

As before, we suppose that $W\in\mathbb{N}$ workers are asked to provide labels to a series of $M$-class classification tasks whose ground truths $g_t,t=1,\ldots,T$ are unknown. We will use a one-hot encoding of the ground truths, i.e.,  $g_t\in\mathbb{R}^M$ will be expressed as $g_t\in\{[1,0,0,\ldots,0]^T,[0,1,0\ldots,0]^T,\ldots,[0,0,\ldots,0,1]^T\}\in\mathbb{R}^M$.


We will associate a skill level with every worker using a homogeneous Dawid-Skene model, where each worker is assumed to have the same accuracy and error probabilities on each class. Formally, worker $i$ provides label $l\in\mathbb{R}^M$ with probability 
\[\left\{
\begin{array}{cc}
 \mathbb{P}(Y_{i,t}=l)) = p_i   & \text{if~}l=g_t \\
 \mathbb{P}(Y_{i,t}=l)) =\frac{1-p_i}{M-1}     & \text{if~}l\neq g_t. 
\end{array}
\right.\]

Similar to binary tasks, \citet{Li2011} showed that the optimal prediction  method under homogeneous Dawid-Skene model is  weighted majority voting. More specifically, when $p_i,\forall i=1,\ldots,W$ are known, the oracle MAP rule is
\[
\gamma_{s,A}^*(Y)=\argmax_{l\in[M]}\sum_{i:(i,t)\in A}v_i^*\one{Y_{i,t}=l},
\]
where $v_i^*=\log\frac{(M-1)p_i}{1-p_i},\forall i\in[W]$. The proof of this can be obtained by following the same line  as in Section~\ref{Sec:twostep}.

In order to construct the weighted majority voting model, we extend PGD algorithm to handle multi-class tasks by showing the skill estimation problem is still a rank one matrix completion problem as follows.
\begin{lemma}
Let us define skill levels $$s_i= \frac{M}{M-1}p_i-\frac{1}{M-1},$$ and noisy covariances
$$\tilde{C}_{ij}=\frac{1}{N_{ij}}\sum_{ t ~|~ (i,t),(j,t)\in A}\langle Y_{i,t},Y_{j,t}\rangle.$$ Then 
$$ E \left[ \frac{M-1}{M} \tilde{C} - \frac{1}{M-1} \right] = s s^T.$$
\end{lemma}
\begin{proof}
Since the random vectors $Y_{i,t}$ and $Y_{j,t}$ are independent, we can write the expectation of the inner product of $Y_{i,t}$ and $Y_{j,t}$ as
\[
\mathbb{E}[\langle Y_{i,t},Y_{j,t}\rangle]=p_ip_j+\frac{1}{M-1}(1-p_i)(1-p_j).
\]
This follows because the inner product is one only if $Y_{i,t}=Y_{j,t}=l$, for some label $l$, and the probability of this is either $p_ip_j$ or $\frac{1-p_i}{M-1}\frac{1-p_j}{M-1}$ depending on whether $l=g_t$ or $l\neq g_t$. 

A simple algebraic manipulation gives the following 
\[
\mathbb{E}[\langle Y_{i,t},Y_{j,t}\rangle]=\frac{M-1}{M}\left(\frac{M}{M-1}p_i-\frac{1}{M-1}\right)\left(\frac{M}{M-1}p_j-\frac{1}{M-1}\right)+\frac{1}{M}.
\] which implies 
\[ E \left[ \frac{M}{M-1} \langle Y_{i,t},Y_{j,t}\rangle  - \frac{1}{M-1} \right] = s s^T. \] We thus have
\begin{eqnarray*}E \left[ \frac{M-1}{M} \tilde{C} - \frac{1}{M-1} \right] & = & \frac{M}{M-1} \frac{1}{N_{ij}} \sum_{t ~|~ (i,t), (j,t) \in A} E \langle Y_{i,t},Y_{j,t}\rangle - \frac{1}{M-1} \\ 
& = & \frac{1}{N_{ij}} \sum_{t ~|~ (i,t), (j,t) \in A} \left( \frac{M}{M-1} E \langle Y_{i,t},Y_{j,t}\rangle - \frac{1}{M-1} \right) \\ 
& = & \frac{1}{N_{ij}} \sum_{t ~|~ (i,t), (j,t) \in A} s s^T \\ 
& = & s s^T.
\end{eqnarray*} \end{proof}

As a consequence of this lemma, if we define 
\[ {\hat C}_{ij} = \frac{M}{M-1} \frac{1}{N_{ij}}  ~\sum_{t ~|~ (i,t), (j,t) \in A}  \langle Y_{i,t},Y_{j,t}\rangle - \frac{1}{M-1}, \] then $s$ can be estimated by solving a rank one matrix completion problem with objective function
\[
\tilde{s}=\argmin_{x\in[-\frac{1}{M-1},1]^W} \sum_{(i,j) \in E}(\hat{C}_{ij}-x_{i}x_j)^2.
\]

As previously, in the limit as $t \rightarrow \infty$, the rank-one problem is an exact match for the problem of recovering skills. In the case where $t$ is finite, we will be in the ``noisy'' regime where $\hat C$ can be thought of as a noise-corrupted version of the true rank-one matrix $s s^T$, with the amount of noise ill decaying to zero as $t \rightarrow \infty$.

\section{Theoretical Results} \label{sec:theory}

Up to now, we have shown how label inference can be reduced to the problem of skill estimation, and addressed the skill estimation problem as a sparse rank-one matrix factorization problem (with noise).
In this section, we analyze which properties of the interaction graph ensure learnability as the number of tasks approaches infinity. Subsequently, we analyze the convergence properties of the PGD algorithm for finite tasks. 


\subsection{Learnability}


We start with the analysis for the infinite instance where $C_{ij}=\tilde{C}_{ij}=s_is_j$ (see Eq. (\ref{samp_corr}) for a definition). There are different ways to let the number of tasks approach infinity while keeping an interaction graph fixed.

\medskip

\noindent {\bf Case A:} For a fixed interaction graph $G=([W],E)$ we can consider assignment sets such that the minimum number of shared tasks, $T_{\min}(T) = \min_{(i,j)\in E} N_{ij}(T)$ approaches infinity. Learnability in this context is a property of the interaction graph.

\medskip

\noindent {\bf Case B:}
We can also consider an  infinite assignment set $A$ and define $G_A^\infty = ([W],E)$ as the graph where two workers are connected by an edge if $N_{ij}=\infty$. In other words, we define connectivity based on whether two workers interact finitely or infinitely many times. 

\medskip

We will follow the second approach as it is slightly more general than the first (the second approach allows assignment sets $A$ where some workers interact only finitely many times, while the first approach does not allow such assignment sets).
Thus, we fix an assignment set $A$, we let $\Theta_A$ be the set of instances sharing assignment set $A$, and we will consider the learnability of subsets $\Theta \subset \Theta_A$. 

To express complete ignorance towards the true unknown labels assigned to tasks, we will consider $\Theta$ which are \emph{truth-complete}: informally, this means that $\Theta$ places no constraints on what the ground truth could be. Formally, truth completeness means that, for any $\theta = (s,A,g)\in \Theta$, we require $\Theta_{s,A}\subset \Theta$ where
$\Theta_{s,A} = \{ (s,A,g) \,: g\in \{-1,+1\}^{\mathbb{N}} \}$. Truth-completeness expresses that there is no prior information about the unknown labels.

As discussed before, the inference problem is inherently symmetric:
the likelihood assigned to some observed data $Y$ under an instance $\theta = (s,A,g)$ is the same as under the instance $(-s,A,-g)$. Thus, an instance set cannot be learnable unless somehow these symmetric solutions are ruled out. 

To express the condition this forces us to adopt will require a few more definitions.
In particular, given $\Theta$ we let $S(\Theta) =  \{ s\in [-1,1]^W\,:\, (s,A,g)\in \Theta \}$ be 
the set of skill vectors that are present in at least one instance in $\Theta$.
For a skill vector $s\in [-1,1]^W$ we let $P(s) = \{ i\in [W]\,:\, s_i>0 \}$ be the set of workers
whose skills are positive and we let 
$\mathcal{P}(s) = \{ P(s), P(-s) \}$ be the (incomplete) partitioning of workers into workers with positive and negative skills; note that workers with zero skill are left out.

With these definitions in place, we will say that $\Theta$ is \emph{rich} if there exists $s\in [-1,1]^W$ and $\alpha>1$ such that
$\times_{i\in [W]} \{ \alpha s_i, s_i/\alpha \}\subset S(\Theta)$ (in other words, there must exist $s \in S(\Theta)$ and $\alpha > 1$ such that we can scale each component of $s$ by either $\alpha$ or $1/\alpha$ and remain in $S(\Theta)$). This is a fairly mild condition; it is satisfied if, for instance, there is some point in $s \in S(\Theta)$ such that a small open-set around that point that is fully contained in $\Theta$.

Richness is required so that there is sufficient ambiguity about skills. Indeed, if richness is not satisfied, then  either every skill vector in $S(\Theta)$ is a spammer or hammer (i.e., $s_i \in \{-1,1\}$ for all $i$) or $S(\Theta) \cap (-1,1)^W$ has a specific structure.  This structure could potentially be exploited by an algorithm. One might say that assuming richness requires an algorithm to be agnostic to any specific structural knowledge of skill vectors. 

We are now ready to state our first main result, which characterizes when rich, truth-complete sets are learnable. 

\begin{theorem}[Characterization of learnability] \label{thm:learnability}
Fix an infinite assignment set $A$
and assume that $G = G_A^\infty$ is connected.
Then, a rich, truth-complete set of instances $\Theta \subset \Theta_A$ over $A$ is learnable 
if and only if the following hold:
\begin{enumerate}[(i)]
\item For any $s,s'\in S(\Theta)$ such that $|s| = |s'|$ and $\mathcal{P}(s) = \mathcal{P}(s')$, it follows that $s=s'$;

\label{thm:c1}
\item The graph $G$ is non-bipartite, i.e, it has an odd-cycle.%
\label{thm:c2}
\end{enumerate}
\end{theorem}
Condition \eqref{thm:c1} requires that any $s\in \Theta$ should be uniquely identified by $|s|$ and knowing which components of $s$ have the same sign and which components are zero.
For example, this condition will be met if $\Theta$ is restricted so that it only contains skill vectors that have a positive sum. For an explanation of why such an assumption is needed, see the (boldfaced) remark in Section \ref{sec:prob}. 

We remark that if the graph $G$ is not connected, we can simply apply this theorem to each of its connected components. For example, in the situation where none of the workers $1, \ldots, k$ have shared a task with any of the workers $k+1, \ldots, n$, one could try to simply recover the skills of workers $1, \ldots, k$ from their common tasks and then the skills of $k+1, \ldots, n$  from their common tasks. This allows us to drop the condition in the theorem that $G$ be connected, at the expense of changing (ii) to the assertion that none of the connected components of $G$ should be bipartite.

The forward direction of the theorem statement hinges upon the following result which is proved in Appendix:  
\begin{lemma}\label{lem:asylearn}
For any $g\in \pms$, $s\in [-1,1]^W$ and an assignment set with a connected, non-bipartite interaction graph $G_A^\infty$,
there exists a method to recover $|s|$ and $\mathcal{P}(s)$.
\end{lemma}
The reverse implication in the theorem statement follows from the following result:
\begin{lemma}\label{lem:asylearn1}
Assume that the lengths of all cycles in $G$ are even. Then there exists $s,s'\in [-1,1]^W$, $s\not\in \{-s',s'\}$ such that $C_{ij}=s_is_j = s_i's_j'$.
\end{lemma}
{\bf Learnability for Finite Tasks:} We mention in passing that asymptotic learnability is a fundamental requirement, which if not met precludes any reasonable finite time result. In consequence there is no inference schema $\gamma$ achieves zero regret in this case.


%
%
\subsection{Convergence of the PGD Algorithm}
The previous section established that for learnability 
 the limiting interaction graph $G_A^\infty$ must be a non-bipartite connected graph. 
 We will now show that PGD under these assumptions 
 converges to a unique minimum for both the noisy and noiseless cases. 
 
By the noiseless case, we mean that in the loss $L$ of \cref{lem:equiv}, we set $\tilde C_{ij} = C_{ij}=s_is_j$ for $(i,j) \in E$. That is, we have infinite number of common tasks to estimate $\tilde{C}_{ij}$ which will then \emph{equal} the expected value $C_{ij}$. However, in reality, we always suffer from the estimation error (i.e., $|C_{ij}-\tilde{C}_{ij}|,\forall (i,j)\in E$) which leads to a more troublesome problem than a rank one matrix completion. We also provide an analysis of our PGD algorithm in this "noisy" case.

Our first step is to show that, under the condition $G$ is connected and non-bipartite, the loss has a unique minimum and the PGD algorithm recovers the skill vector. For technical convenience, our theorem below considers recovering the absolute values of skills $|s|$. \ao{This is the same as recovering the vector $s$, as we discuss next.} 

\ao{Indeed, observe that if $A=s s^T$ is rank-$1$, then $|A| = |s| |s|^T$ is also rank-$1$, where the absolute value is taken elementwise. Then we can simply take the absolute value of all-revealed entries; the theorem below will ensure that the PGD method recovers $|s|$. Once $|s|$ is recovered, we need to do some post-processing to recover the sign of each entry. }

\ao{It should be natural that, because we do not assume access to any true labels, recovering $s$ once $|s|$ is available will require some assumption on the vector $s$.  Indeed, even in the simplest scenario of a complete worker-task interaction graph with $W-1$ agents always agreeing with each other and disagreeing with agent $W$, we cannot distinguish between the possibilities that $(s_1=s_2 = \cdots = s_{W-1}=1, s_W = -1)$ and $(s_1 = s_2 = \cdots = s_{W-1} = -1, s_W=1)$. In other words, we fundamentally cannot know if the $W-1$ agreeing agents are lying or telling the truth. Therefore we will be assuming as before that $\sum_{i=1}^W s_i > 0$; this is just saying that there is more truth-telling that lying in the entire system. Under this condition, a post-processing step becomes possible.}

\bigskip

\ao{\noindent {\bf Post-processing for sign recovery.} If $|s|$ is recovered, we recover the signs of each entry as follows. We assign a positive sign to the first worker ($s_1>0$). We then inspect all the elements $s_1 s_j$ over $j$  neighbors of the first worker (i.e., workers that have a joint task with the first worker) and assign a sign to them by inspecting the sign of $s_1 s_j$. We then repeat this, assigning signs to all the neighbors of workers whose sign was just assigned, until the sign of every worker is assigned. Finally, we check if for the resulting vector satisfies $\sum_i s_i > 0$; if not, we flip the sign of every worker. It is immediate that this process always recovers the signs correctly, provided the underlying graph is connected and non-bipartite and $\sum_i s_i > 0$. } 

\bigskip

\ao{\begin{theorem} \label{thm:noiselessPGD}
The PGD Algorithm with $s>0$ and $x(0)>0$ converges to the global minimum $x=s$ under conditions for learnability  of Theorem~\ref{thm:learnability} and small enough stepsize $\eta$. 
\end{theorem}}

As above, if the underlying graph $G$ is not connected, we can, as remarked earlier,  simply apply this theorem to each of the connected components.  The key requirement of Theorem ~\ref{thm:learnability} -- that the underlying graph is not bipartite -- then becomes the requirement that none of the connected components of $G$ are bipartite.

\ao{We next consider the problem of obtaining a polynomial-time convergence rate for the problem, and moreover doing so in the noisy case. Thus we now consider the case when only the perturbed entries $s_i s_j + \Delta_{ij}$ are revealed. We will now make the slightly stronger assumption (discussed at more length below) that the correct answer $s$ lies in the cube ${\cal C} = [\kappa, K]^W$ where $0<\kappa\leq K$. Without loss of generality, we can therefore assume that all revealed entries lie in this set, i.e., 
\[ s_i s_j + \Delta_{ij} \in [\kappa, K] \mbox{ for all } (i,j) \in \Omega,\] because otherwise we can simply threshold the revealed entries over $[\kappa, K]$ while simultaneously reducing the disturbances $\Delta_{ij}$.} 

\ao{It is natural to attempt to generalize our earlier PGD approach to this setting, in particular by doing gradient descent on the ``perturbed'' function 
\[   f_{\Delta}(x) := \frac{1}{2} \sum_{i,j=1}^W N_{ij} (x_i x_j - s_i s_j - \Delta_{ij})^2. \] While this is possible, we pursue a shortcut naturally adapted to this setting, by using a re-scaling of so-called  exponentiated gradient method.} 

\ao{Specifically, defining $\nabla_t$ by 
\[ [\nabla_t]_i = \sum_{j=1}^W N_{ij} (x_i x_j - s_i s_j - \Delta_{ij}), ~~~~~ i = 1, \ldots, W, \] we update as
\begin{equation} \label{expgrad} x(t+1) = P_{\cal C} \left[ x(t) e^{- \alpha \nabla_t} \right]. 
\end{equation}}

\ao{We may think of $\nabla_t$ as related to, but not identical, to the gradient of the perturbed function $\nabla f_{\Delta}(x(t))$. Indeed, observe that the latter quantity will weigh each term in the definition of $\nabla_t$ slightly differently.  Exponentiated gradient methods of this type are common when optimizing over the simplex, where they come from regularization with the KL divergence (see \cite{hazan2016introduction}). They are somewhat less common when optimizing over a cube, as we do here. All the same, the following theorem shows that this method is able to achieve polynomial-time convergence for the perturbed problem.}

\ao{Before we state our main result on the performance of this scheme, we need to introduce some notation}. Note that the condition that $G$ is connected and non-bipartite implies that  the worker-interaction count matrix $N$ is irreducible and aperiodic. The {\em signless Laplacian} matrix is then  defined as  \[ [L_{\rm s}]_{ij} = \begin{cases} N_{ij} & j \neq i \\ 
\sum_{k=1}^n N_{ik} & j=i  \end{cases} .\] 
By contrast, we will use $N$ to denote the matrix whose $i,j$'th entry is $N_{ij}$; the matrix $N$ will thus have zero diagonal. 

The matrix $L_{\rm s}$ contrasts with the usual Laplacian because the off-diagonal elements have positive signs. 
It can be shown that if the graph $G$  is not bipartite, the matrix $L_{\rm s}$ is positive definite \citep{desai1994characterization}.  In fact, the following stronger assertion is true.  We will use $\lambda$ to denote the smallest eigenvalue of the signless Laplacian matrix of a non-bipartite graph with unit weights; we remark that it as consequence of the results of \citet{desai1994characterization} that $\lambda \geq 1/W^3$ (where, recall, $W$ is the number of workers, so that the matrices $N$ and $L_{\rm s}$ are $W \times W$). Finally, we let $N_{\rm min}$ be the smallest positive weight among $\{N_{ij}\}$. \ao{Our final main result, which obtains a polynomial-time convergence rate in both the unperturbed and perturbed cases, is given in the following theorem.}

\ao{\begin{theorem} Suppose $s$ is located in the interior of $[\kappa, K]^W$ where $0 < \kappa \leq K$. Provided $\max_{i,j} |\Delta_{ij}|$ is small enough and $\alpha = (2 \sqrt{W} ||N||_2 K^2)^{-1}$, we have that:
\begin{enumerate} \item  Eq. (\ref{expgrad}) has a limit, which we will denote by $x_{\Delta}^*$. \item Convergence to any neighborhood of $x_{\Delta^*}$ occurs in polynomial-time. \item $x_{\Delta}^*$ is close to $s$ in the following sense:  
\[ ||x_{\Delta}^* - s||_2 \leq  K \frac{\sqrt{W} ||N||_{\infty}}{\mu} \max_{i,j} |\Delta_{ij}|, \] with $$\mu = \kappa^2 \lambda_{\rm min}(L_s) N_{\rm min}.$$  In particular, in the noiseless case when $\Delta=0$, we have that $x_0^*=s$.
\end{enumerate} \label{thm:pert}
\end{theorem}} 

\bigskip

\ao{We note that the assumptions of Theorem \ref{thm:noiselessPGD} are slightly weaker than the assumptions of Theorem \ref{thm:pert}. While the former assumed that $s>0$, the latter assumed the slightly stronger statement that $\min_i s_i > \kappa > 0$. This can always be accomplished by throwing out nodes with $s_i \approx 0$ from the data set; such modes are making random guesses and do note contribute to the accuracy of the Bayes classifier of Eq. (\ref{eq:plugin}) which should assign them zero weight. A natural way to do this is to simply set to zero any $N_{ij}$ corresponding to correlations  $\widetilde{C}_{ij}$ whose absolute values are smaller than $\delta + O(\sqrt{(\log W)/T})$ for some small $\delta>0$. The advantage of this threshold is that all agents with $s_i=0$ will, with high probability, have all their interactions $N_{ij}$ set to zero and thus automatically ignored by both the PGD method and the variant of exponentiated gradient proposed here.  On the other hand, any pair of workers $i,j$ with $|s_i|$ and $|s_j|$  strictly larger than $\sqrt{\delta}$ will have their correlation above this threshold with high probability .} 

\aor{Finally, we discuss the key ``trick'' underlying the proof of this theorem. The main idea is to interpret the update of Eq. (\ref{expgrad}) as a projected gradient descent on the function 
\[ g_{\Delta}(z) = \frac{1}{2} \sum_{i,j=1}^W N_{ij} e^{z_i + z_j} - \sum_{i=1}^W  z_i \sum_{j=1}^n N_{ij} (s_i s_j + \Delta_{ij}),\] after a change of variable. The construction of this function is what allows us to bypass a lot of the technical difficulties in the analysis. }

\bigskip

{\bf \textsc{Finite-Task Bound:}} Note that we can directly apply this result to obtain a finite task characterization as well. In particular consider a connected and non-bipartite interaction graph. Define $d_{\max}$ as the maximum degree and $D$ as the sum of the degrees. It follows by standard Hoeffding bounds that with probability greater than $(1-\delta)$ we have $\max_{(i,j)\in E} |C_{ij} - \hat C_{ij}| \leq \frac{\log(D/\delta)}{\sqrt{N_{\min}}}$. \ao{We can then set $\Delta =C_{ij} - \hat C_{ij}$ and plug this bound into the above theorem to obtain 
\[ ||x_{\Delta}^* - s||_2 \leq  K \frac{\sqrt{W} ||N||_{\infty}}{\mu} \frac{\log(D/\delta)}{\sqrt{N_{\min}}}. \]}

\section{Experimental Results\label{sec:exp}}

In this section, we will be showing the experimental results to the PGD scheme. We will make one minor modification to the algorithm by adding a projection away from the boundary of the cube $s_i=1$ by projecting $x_i$ onto  $[-1+t/\sqrt{N_i},1-\tau/\sqrt{N_i}]$ at every step, where recall $N_i$ is the number of tasks assigned to agent $i$ and $\tau$ is a parameter. The justification is that skills close to one or negative one have an overwhelming impact on the plug-in rule of Eq. (\ref{eq:plugin}). According to Hoeffding inequality, the skill estimates are expected to have an uncertainty proportional to $\tau/\sqrt{N_i}$ with probability $\mathrm{const}\times e^{-\tau^2}$.
 There is little loss in accuracy in confining the parameter estimates to the appropriately reduced hypercube,
 while in principle one could tune this parameter, we use $\tau=1$ in this paper.

\begin{figure*}
\centering
\subfigure{\label{subFig:clique}}\addtocounter{subfigure}{-1}
\subfigure[Clique.]{
\includegraphics[width=0.31\textwidth]{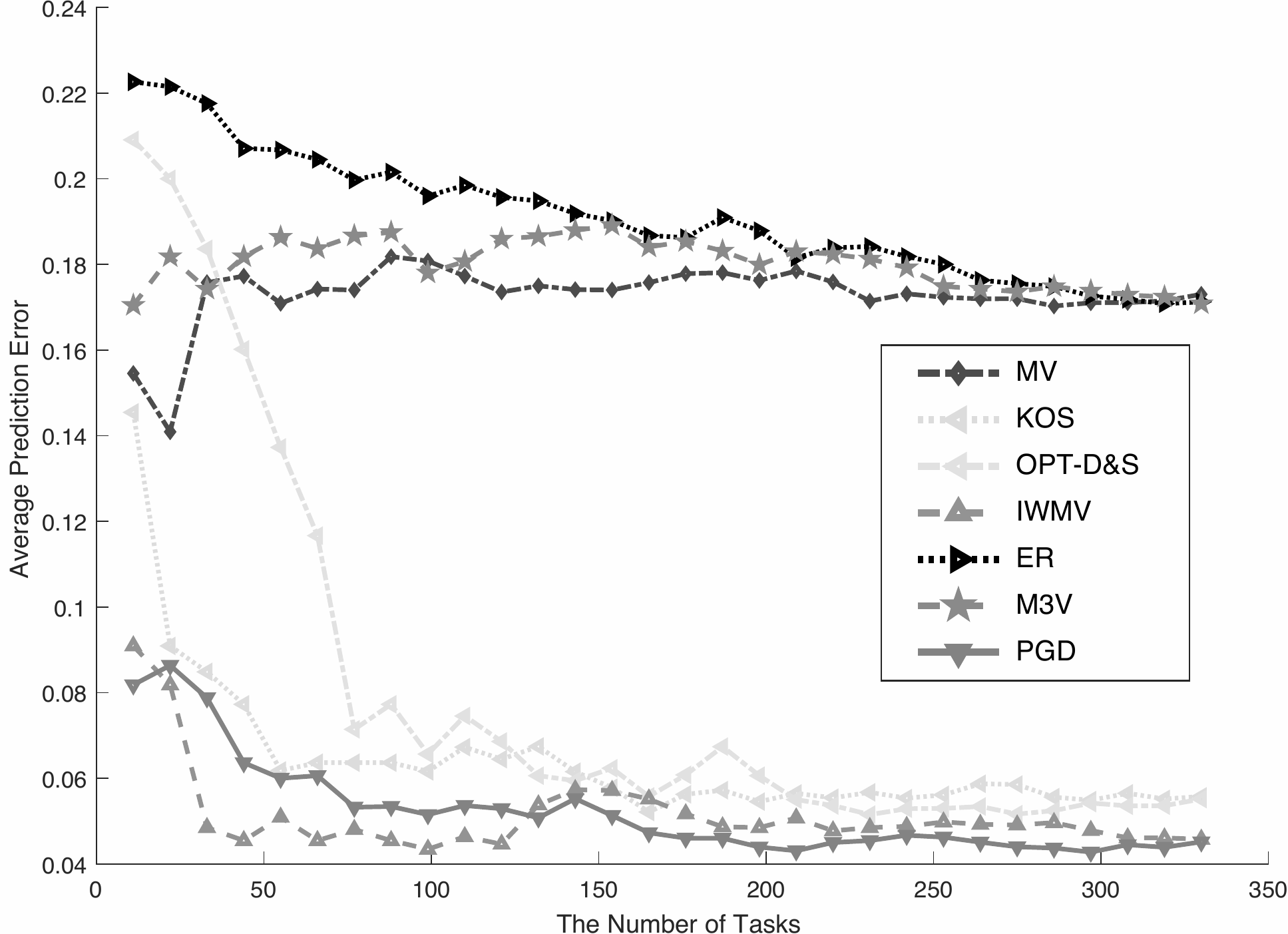}}
\subfigure{\label{subFig:star}}\addtocounter{subfigure}{-1}
\subfigure[Star graph with $3$-cycle.]{
\includegraphics[width=0.31\textwidth]{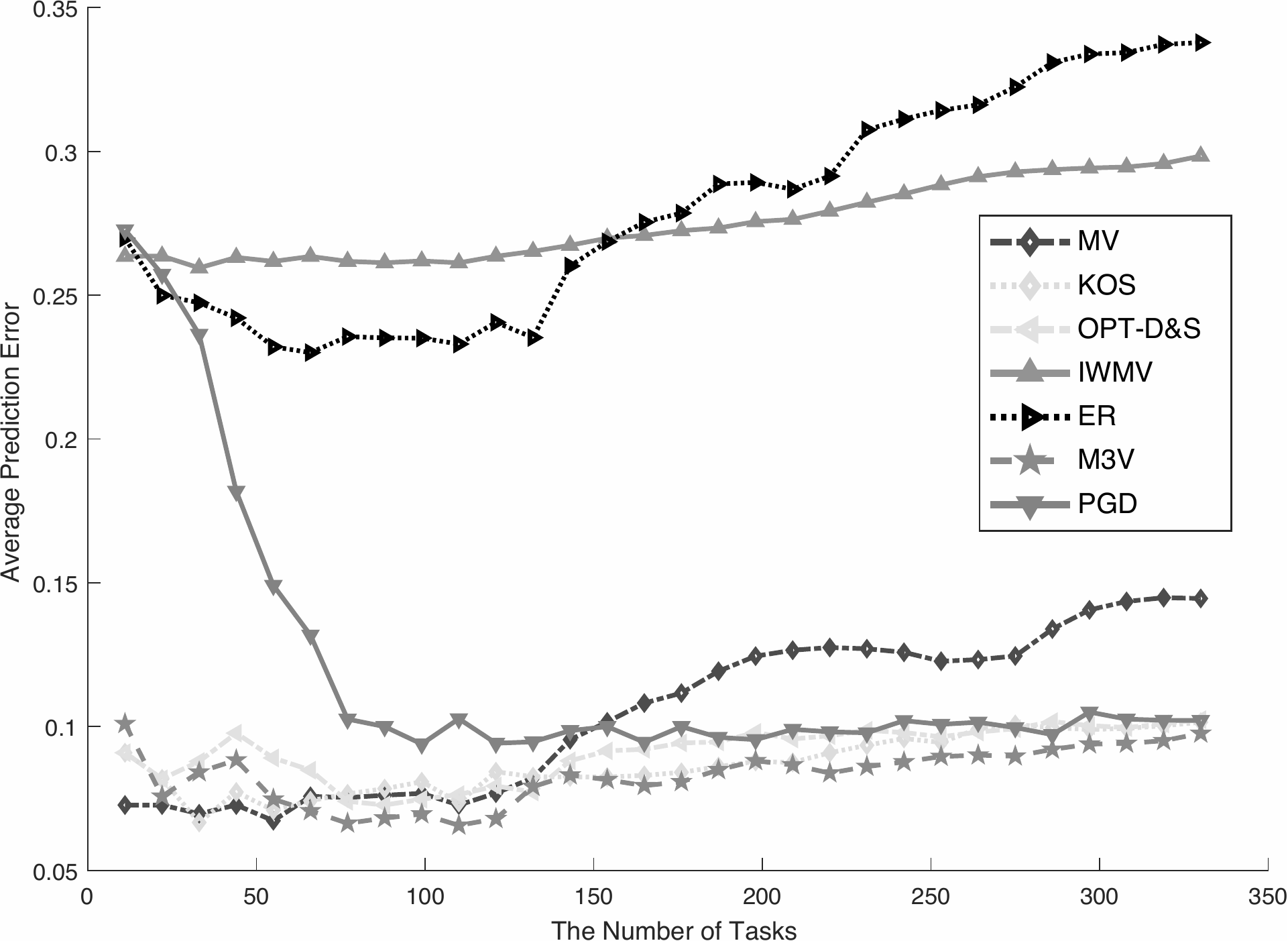}}
\subfigure{\label{subFig:ring}}\addtocounter{subfigure}{-1}
\subfigure[Ring.]{
\includegraphics[width=0.31\textwidth]{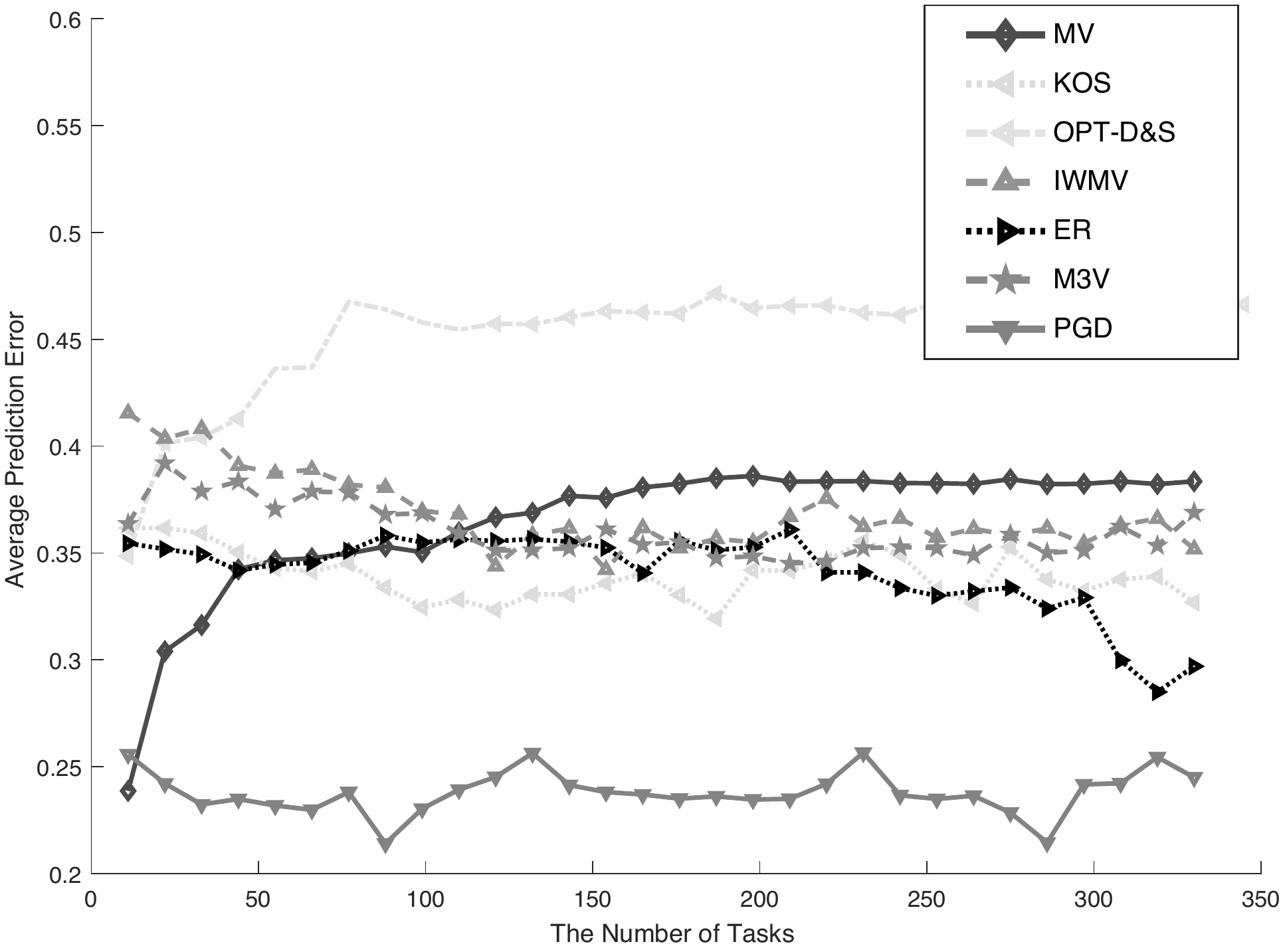}}
\label{Fig:comparison}
\vspace{-0.15in}
\caption{\small Illustrative comparisons of prediction performance for three graph types. Only mean values are plotted for exposition. For the clique, standard deviation values with 11 tasks were $0.09$, $0.10$, $0.14$, $0.14$, $0.19$, $0.09$, and $0.09$ for MV, KOS, OPT-D\&S, PGD, ER, IWMV, and M3V respectively; and with 330 tasks they were $0.02$, $0.01$, $0.017$, $0.014$, $0.012$, $0.013$, and $0.018$ respectively. For the star-graph the standard deviations for 11 tasks were $0.09$, $0.13$, $0.13$, $0.06$, $0.13$ $0.09$, and $0.07$ for MV, KOS, OPT-D\&S, PGD, ER, IWMV, M3V respectively and for 330 tasks they were $0.016$, $0.015$, $0.013$, $0.012$, $0.04$, $0.03$, and $0.013$. For the ring the standard deviation for 11 tasks were $0.096$, $0.05$, $0.08$, $0.1$, $0.11$, $0.09$, $0.08$ and for 330 tasks they were $0.017$, $0.05$, $0.02$, $0.043$, $0.05$, $0.086$, $0.05$. Standard deviations decrease with growing number of tasks.}
\end{figure*}
{\bf \textsc{Synthetic Experiments:}}
We will experiment with different graph types, increasing levels of label noise, graph-size, skill distribution, and different weighting functions on synthetic data. 

{\it Impact of Graph Type:}
We consider three 11-node (\# workers) irreducible, non-bipartite graphs, namely, a Clique ($G_1$), Star with augmented odd cycle ($G_2$), and a Ring ($G_3$) to illustrate the impact of sparsity (Clique has dense worker interactions while Star/Ring have fewer than 3 worker interactions) and graph-type (Ring vs. Star). An illustration of the different graph types is shown in Figure~\ref{Fig:GraphType}. These graphs satisfy condition (ii) of Thm~\ref{thm:learnability}. 

\begin{figure*}
\centering
\subfigure[Clique.]{
\includegraphics[width=0.27\textwidth]{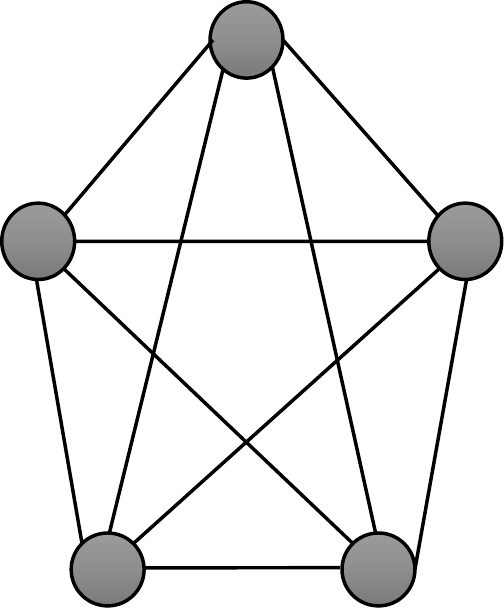}}
\subfigure[Star graph with $3$-cycle.]{
\includegraphics[width=0.27\textwidth]{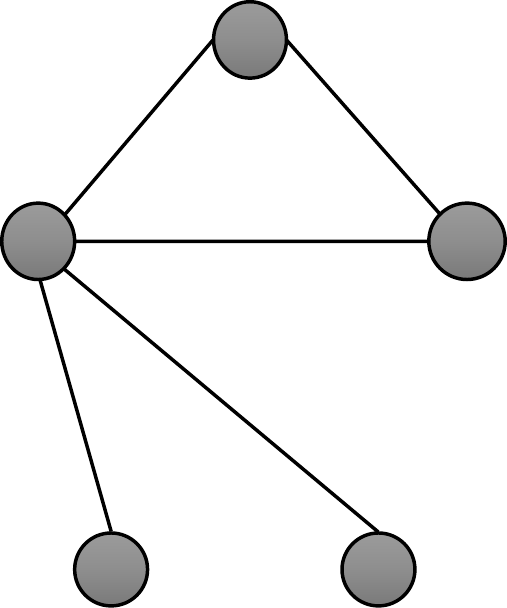}}
\subfigure[Ring.]{
\includegraphics[width=0.27\textwidth]{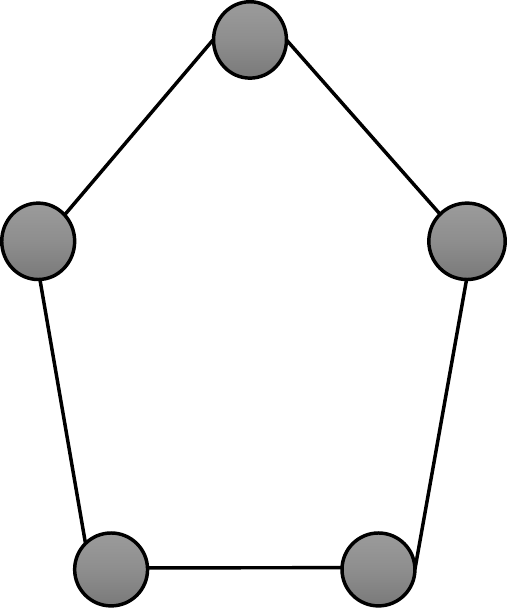}}
\caption{Different Graph Types.}\label{Fig:GraphType}
\end{figure*}

{\it Noise Robustness:} 
To see the impact of noise, we vary the noise level by increasing the number of tasks, which in turn reduces the error in the correlation matrix. 
Tasks are randomly assigned to binary classes $\pm{1}$ with total number of tasks ranging from $11$ to $330$. Skills are randomly assigned on a uniform grid between $0.8$ and $-0.3$\footnote{
The reason for this choice is to satisfy condition (i) in \cref{thm:learnability}, i.e., requiring overall skills to be positive. Aggregate skill is about $0.25$. 
} 

{\it Influence of  Different Skill-Distribution:}
We randomly assign binary classes to $T=300$ tasks and select five pairs of parameters. Average prediction errors are presented in Table~\ref{Tab:beta} averaged over $10$ independent runs. Parameters $\alpha=5,\beta=1$, correspond to reliable workers leading to small prediction error; the prediction error with parameters $\alpha=2,\beta=2$ and $\alpha=0.5,\alpha=0.5$, is almost random because of $\sum_{i\in [W]}s_{i}$ is no longer positive, which validates our theory. Similar situation arises for $\alpha=2,\beta=5$ and $\alpha=5,\beta=1$, because the skills are all flipped relative to our assumption that the sum of the skills is positive. 

\iftrue
\begin{table}
\centering
\scalebox{1}{
\begin{tabular}{c|c|c|c|c}
\hline
Type of workers                & $\alpha$  & $\beta$   & Bayes error           & Prediction error (const. noise) \\  \hline
Adversary vs. hammer            & $0.5 $    & $0.5$     &$ 0.0036\pm 0.0014$    &$ 0.5990\pm 0.4860$ \\ \hline
Asym. with more positive skills & $5 $      & $1$       &$ 0.0038\pm 0.0014$    &$ 0.0041\pm 0.0013$ \\ \hline
Asym. with more negative skills & $2 $      & $5$       &$ 0.0314\pm 0.0062$    &$ 0.9667\pm 0.0067$ \\ \hline
Hammer                          & $2 $      & $2$       &$ 0.0615\pm 0.0083$    &$ 0.4162\pm 0.4273$ \\ \hline
Spammer                         & $1 $      & $3$       &$ 0.0129\pm 0.0034$    &$ 0.9864\pm 0.0041$ \\ \hline
\end{tabular}
}
\caption{Average prediction errors with different skills distributions.}
\label{Tab:beta}
\end{table}
\fi

{\it Influence of Graph Size:}
We focus on how the graph size affects the performance of PGD algorithm. Note that graph size is associated with the number of workers. Our goal is to demonstrate that for a constant amount of noise, prediction accuracy of PGD does not degrade with graph-size. We again consider the case when the worker-interaction graph is a star-graph with an odd-cycle of length 3. We increase the size of worker-interaction graph by adding nodes to the star-graph. Skills $s$ are selected between $0.8$ and $-0.3$ uniformly. 
To fix the noise level, we define $C_{ij}=s_is_j+\xi_{ij},\forall (i,j)\in E$ where $\xi_{ij}$ is randomly selected from $[-0.2,0.2]$. Note that the noise level is quite large relative to what we expect in terms of accuracy of correlation estimates. We iteratively run PGD for $50$ times. The average prediction errors with different graph size it presented in Table~\ref{Tab:graphsize}. It can be seen that the prediction error is not sensitive to the graph size compared to the Bayes error.

\iftrue
\begin{table}
\centering
\scalebox{0.85}{
\begin{tabular}{c|c|c|c|c}
\hline
  Number of workers & $21$  & $51$ & $71$ & $91$  \\ \hline
Bayes error & $0.0425\pm{0.0042}$  & $0.0622\pm{0.0040}$ & $0.0634\pm{0.0033}$ & $0.0574\pm{0.0030}$ \\ \hline
Prediction error (const. noise)& $0.0425\pm{0.0042}$ & $0.0641\pm{0.0126}$ & $0.0662\pm{0.0072}$ & $0.0618\pm{0.0063}$\\ \hline
\end{tabular}}
\caption{Average prediction errors for different graph sizes.}
\label{Tab:graphsize}
\end{table}
\fi
{\it Influence of Weighting function:}
It is straighforward from our proof of Theorem~\ref{thm:noiselessPGD} to see that PGD algorithm converges to the global optimal for any non-negative weights. Our objective is based on weighting with number of counts in Eq.~\ref{lem:equiv}. However, there are other options that one could consider.  \cite{dalvi_aggregating_2013} has suggested using $B(N_{ij})=N_{ij}^2$, while we use $N_{ij}$. Another possibility is to use binary weights. 
\begin{table*}
\caption{Prediction errors for different weightings}
\label{Tab:weighting}
\centering
\begin{tabular}{c|c|c|c}
\hline
Worker type  & & &\\ Assigned most tasks &  $[N_{ij}>0]$  & $B(N_{ij})=N_{ij}$ & $B(N_{ij})=N_{ij}^2$\\ \hline
Spammers & $0.33\pm{0.03}$  & $0.33\pm{0.03}$ & $0.55\pm{0.17}$\\ \hline
Positive skill workers& $0.17\pm{0.06}$ & $0.09\pm{0.02}$ & $0.09\pm{0.02}$\\ \hline
\end{tabular}
\end{table*}
We iteratively run PGD $10$ times for each weighing function with $T=300$ tasks for different types of task assignements. If $N_{ij}$'s are all equal, these choices produce identical results. We consider two cases: (a) Spammers are assigned a majority of tasks; (b) Positively skilled workers are assigned most tasks. The prediction errors are compared in Table~\ref{Tab:weighting}. Note that quadratic weighting is quite bad in this case because it tends to ignore positively skilled workers. On the other hand unweighted case does not accurately estimate spammers and also results in poor choice.

We compare the average prediction error $PE=\frac{1}{T}\sum_{t=1,\ldots,T}1\{\hat{Y}_t\neq g_t\}$ with the Majority Voting (MV) algorithm, the KOS algorithm \citep{Karger2013},  Opt-D\&S algorithm \citep{Zhang2014}, the ER algorithm \citep{dalvi_aggregating_2013}, the IWMV algorithm \citep{Li2011}, and the M3V algorithm \citep{tian2015}. The KOS algorithm is based on belief propagation, Opt-D\&S uses a spectral method to initialize EM, the ER algorithm is the more successful spectral method of the paper defining it, the IWMV algorithm is an EM-style algorithm.
Each algorithm is averaged over 15 trials on each dataset. The average prediction errors are presented in Figure~2. As the number of tasks grows, the average prediction error of PGD algorithm decreases. PGD is evidently robust to missing data/sparsity and graph-type. OPT-DS, which is close to PGD performance suffers significant performance degradation on sparse graphs such as rings. We can attribute this to the fact that a tensor-based method requires at least 3 worker annotations for each task ~\cite{Zhang2014}. 
\begin{table*}
\caption{Summary of Benchmark Datasets.}
\label{Tab:realdata}
\centering
\begin{tabular}{c|c|c|c|c|c}
\hline
Datasets & Tasks & Workers & Instances & Classes & Sparsity level\\ \hline
RTE1     & 800  &  164    & 8000    &   2 &0.0610\\ 
Temp     & 462   & 76      &4620      &   2 &0.1316\\ 
Dogs      & 807   & 109     & 8070    &   4&0.0917\\ 
WSD      & 177    & 34      & 1770    &   3 & 0.2947\\
Web      & 2665    & 177      & 15567    &   5&0.0033\\ 
\hline
\end{tabular}
\end{table*}
\begin{table*}
\caption{ Prediction Errors of Different Methods.}
\label{Tab:errorrate}
\centering
\begin{tabular}{c|c|c|c|c|c|c|c}
\hline
Data& PGD  & MV & Opt-D\&S  & KOS  &  ER & IWMV &M3W \\ 
\hline
RTE1  &\textbf{0.07} & 0.1031        & 0.0712       &  0.3975  & 0.14 &0.08 & 0.0813\\ 
Temp   &\textbf{0.0512}& 0.0639       &  0.0584        & 0.0628  & 0.052&0.06& 0.0606\\ 
Dogs   & \textbf{0.1660}& 0.1958  &  0.1689          & 0.3172  & 0.18&0.19 &0.1822\\
WSD    &         0.0056 & 0.0056 & N/A            &0.0056   & 0.0056&0.0056&0.0056\\ 
Web    &\textbf{0.1485}& 0.2693  &  0.1586             & 0.4293  & 0.22&0.22&0.1847\\ 
\hline
\end{tabular}
\end{table*}
{\bf \textsc{Benchmark Dataset Experiments:}}
We illustrate the performance of PGD algorithm against state-of-art algorithms. Each algorithm is executed on five data-sets, i.e. RTE1 \citep{Snow08}, Temp \citep{Snow08}, Dogs \citep{Deng2009}, WSD (Word Sense Disambiguation) \citep{Snow08},and WebSearch \citep{Zhou2012}.  A summary of these data-sets is presented in Table~\ref{Tab:realdata}. Following convention we report errors between ground-truth and recovered labels in Table~\ref{Tab:errorrate}. Note that on WSD dataset, OPT-D\&S algorithm does not converge to a equilibrium point after $1000$ iterations.

\begin{figure} 
\centering
\subfigure{\label{subFig:rte}}\addtocounter{subfigure}{-1}
\subfigure[RTE1.]{
\includegraphics[width=0.4\textwidth]{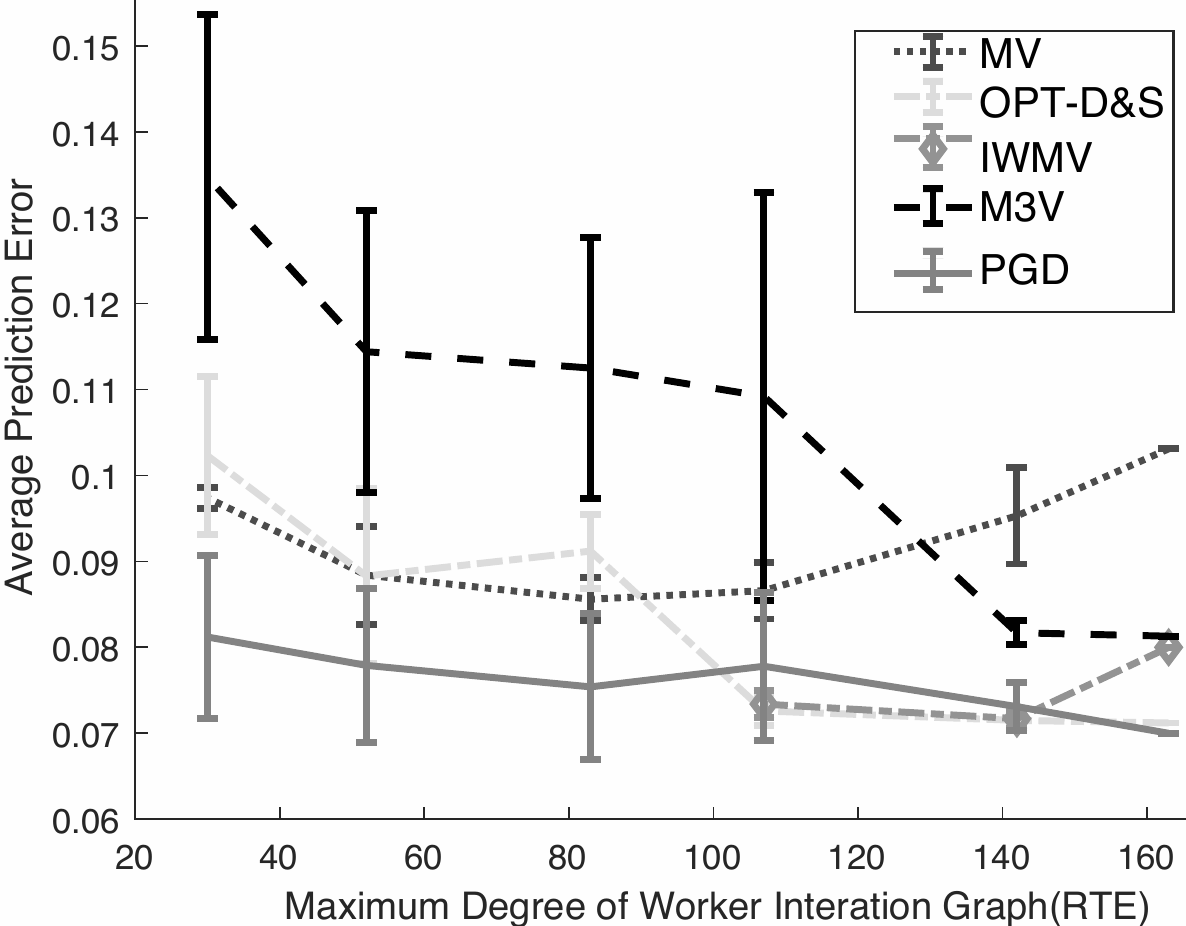}}
\subfigure{\label{subFig:temp}}\addtocounter{subfigure}{-1}
\subfigure[TEMP.]{
\includegraphics[width=0.4\textwidth]{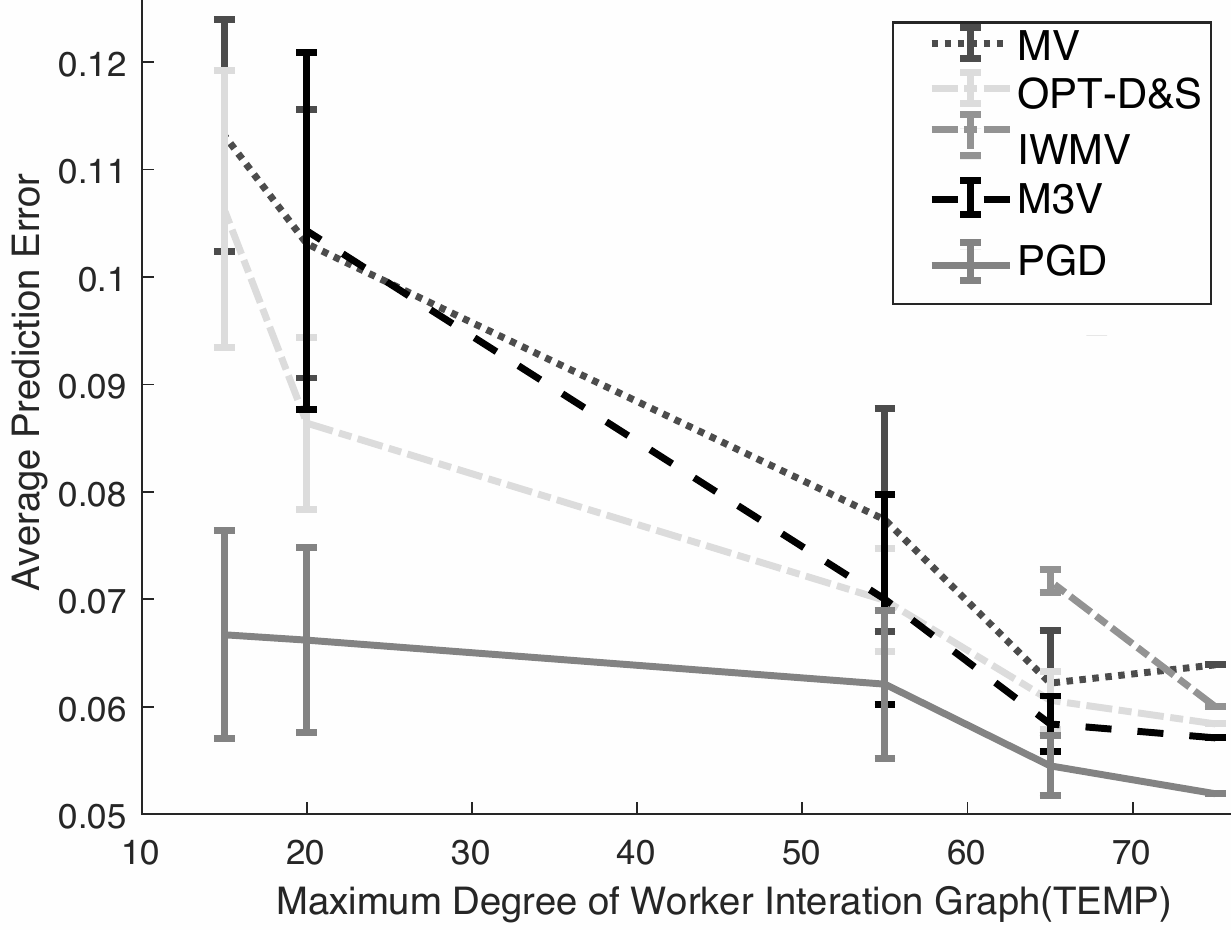}}
\label{Fig:Sparsity}
\caption{\small Impact of Graph Sparsification.}
\end{figure}

{\it Influence of Graph Sparsification:}
Here we consider the scenario where fewer workers label each task on the binary classification benchmark datasets. Binary classification tasks are aligned with our theoretical results. This experiment will highlight the performance of state-of-art algorithms under sparse task-assignments. We simulate this effect based on random sparsification. 
In particular, we sort the degree of each node on the interaction graph. To sparsify the graph we randomly delete edges starting with the highest degree node and continue this process for other nodes until we obtain an interaction graph with desired maximum degree. We also remove symmetrically remove corresponding edges of incident workers to maintain symmetry. This has the implicit effect of deleting some of the tasks as well (for instance, if a task is annotated by two workers). Higher levels of sparsification leads to fewer availability of tasks for training. We iteratively run PGD  and the other algorithms for $50$ Monte-Carlo trials with different desired maximum degrees. The average prediction errors are displayed in Figure~3. The reason IWMV performs poorly is that majority votes are no longer reliable, which IWMV relies on. Our PGD algorithm is surprisingly robust to sparsification of interactions and degrades gracefully relative to other schemes. This highlights the fact that PGD is capable of leveraging sparse interactions among workers and obtain fairly robust estimates of skill-levels required for accurate prediction. 

\bigskip

\ao{ \noindent {\bfseries Normal gradient descent vs Eq. (\ref{expgrad}):} Although we have found it convenient to use Eq. (\ref{expgrad}) for our polynomial-time convergence results, in practice we do not see much advantage of that iteration compared to normal gradient descent. Figure \ref{fig:expvsnormal} gives a comparison in the noiseless case while Figure \ref{fig:expvsnormalpert} gives a comparison in the noisy case. The results are extremely similar. Results are shown for random $s$ and three choice of graphs; each plot is the result of a single realization. In general, all the realizations we have seen look like the plots shown, with only minor differences between the methods. }

 \begin{figure}[htp]

\centering
\includegraphics[width=.3\textwidth]{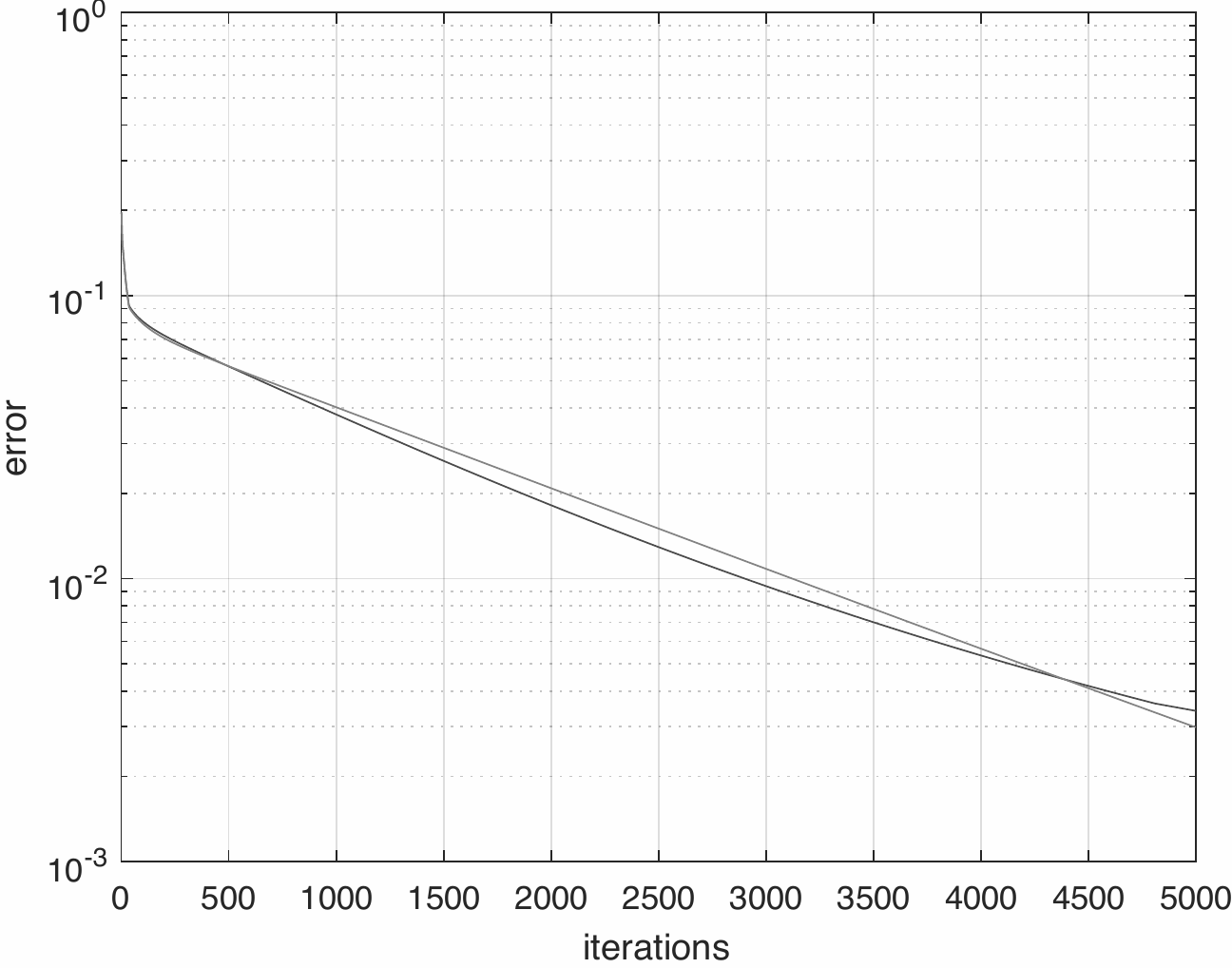}\hfill
\includegraphics[width=.3\textwidth]{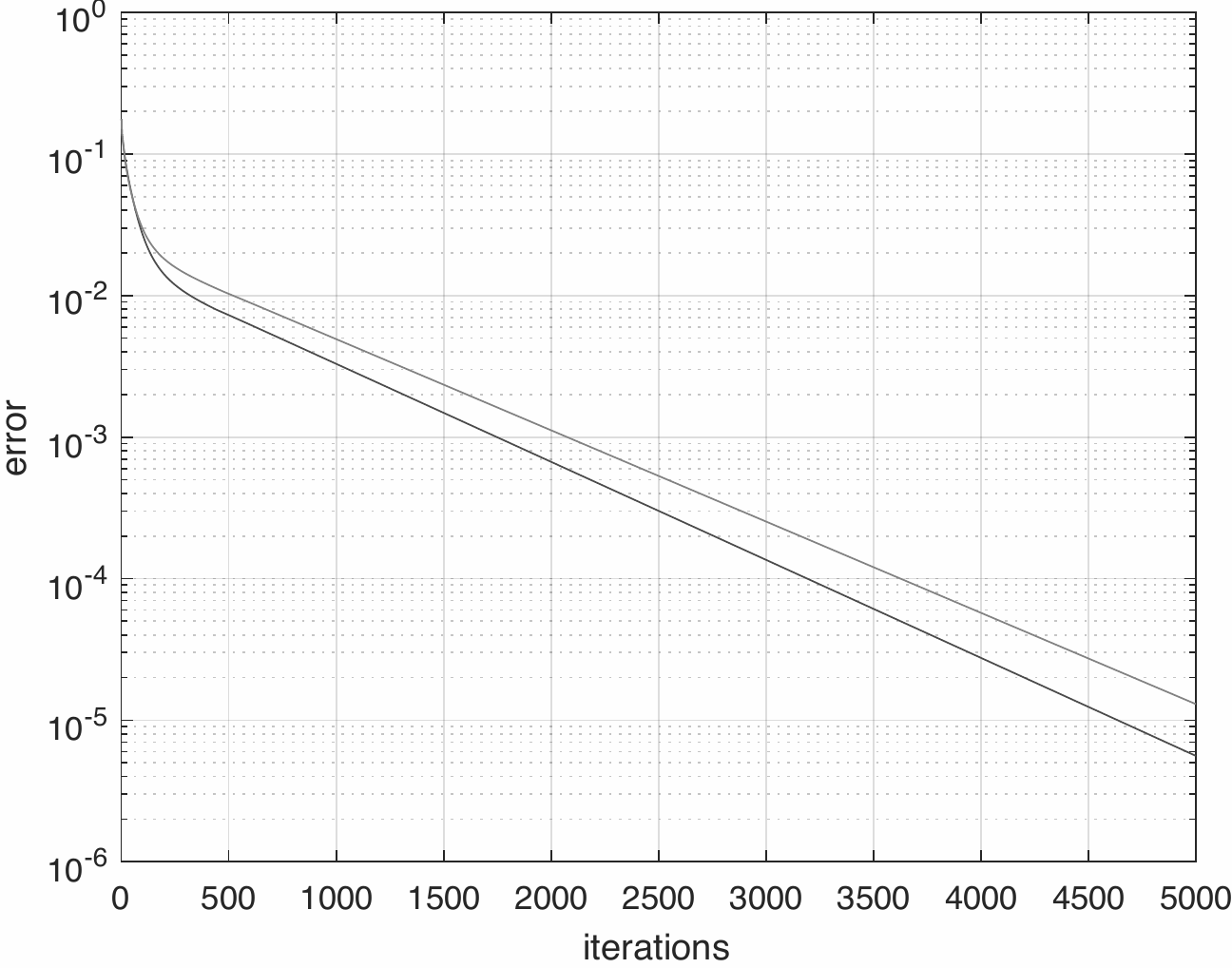}\hfill
\includegraphics[width=.3\textwidth]{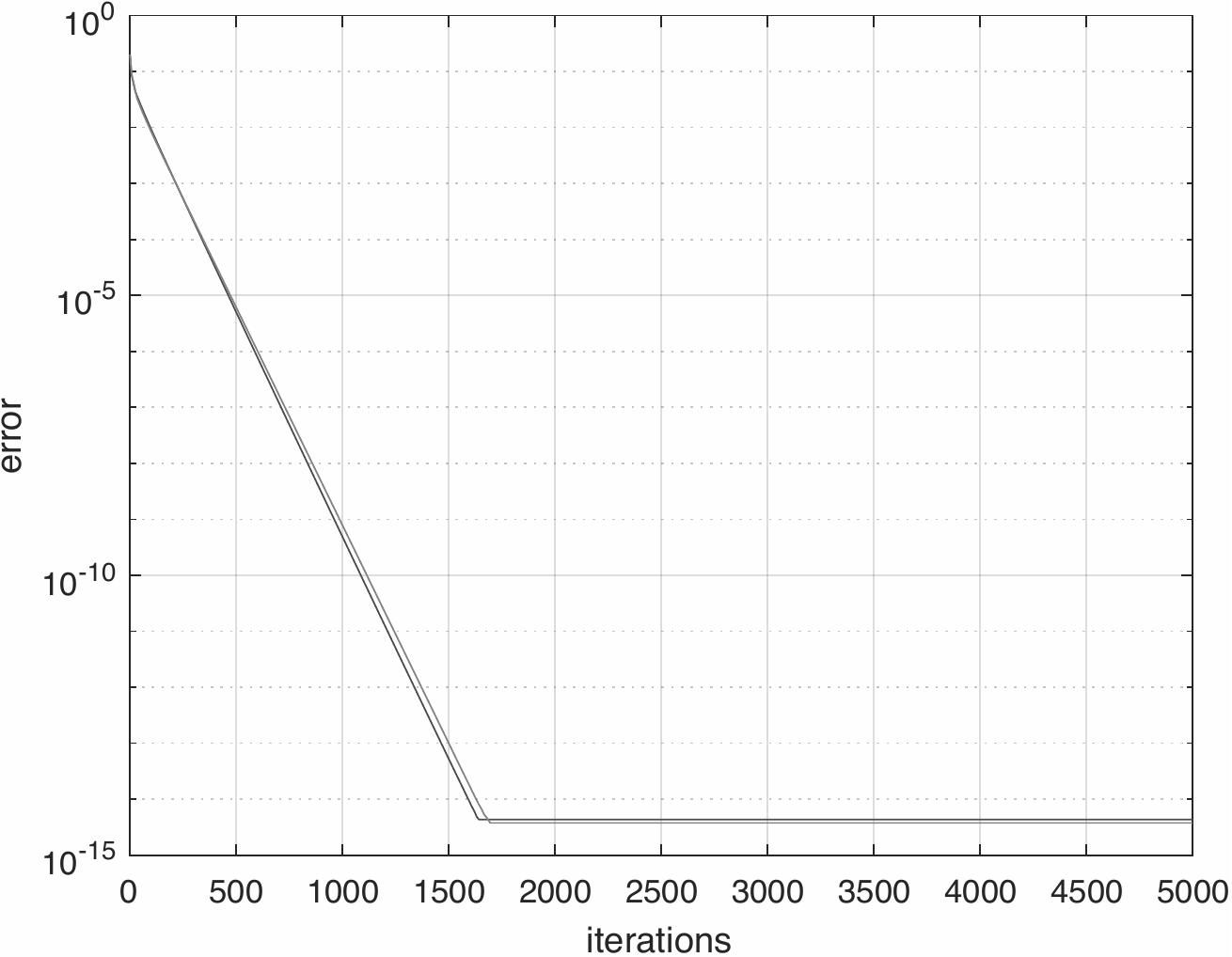}

\caption{\ao{{\bfseries Noiseless case: } the figures show gradient descent in blue vs the variant of exponentiated gradient we have proposed in Eq. (\ref{expgrad}) in red. The plots show the error $||x(t) - s||_{\infty}$ on the y-axis vs the iteration $t$ on the x-axis. The circle graph is shown on the left, the 2D grid with one extra edge in the middle, and an Erdos-Renyi random graph on the right. On all cases, the number of workers is $25$. A single edge was added to the 2D grid in order to ensure aperiodicity. All the perturbations are zero and the correct answer $s$ was generated with each component uniformly random over $[0.4,0.6]$. Starting point was $x(0)=0.6 \cdot {\bf 1}$ and stepsize was taken to be $0.01$ in both cases. }}
\label{fig:expvsnormal}

\end{figure}

 \begin{figure}[htp]

\centering
\includegraphics[width=.3\textwidth]{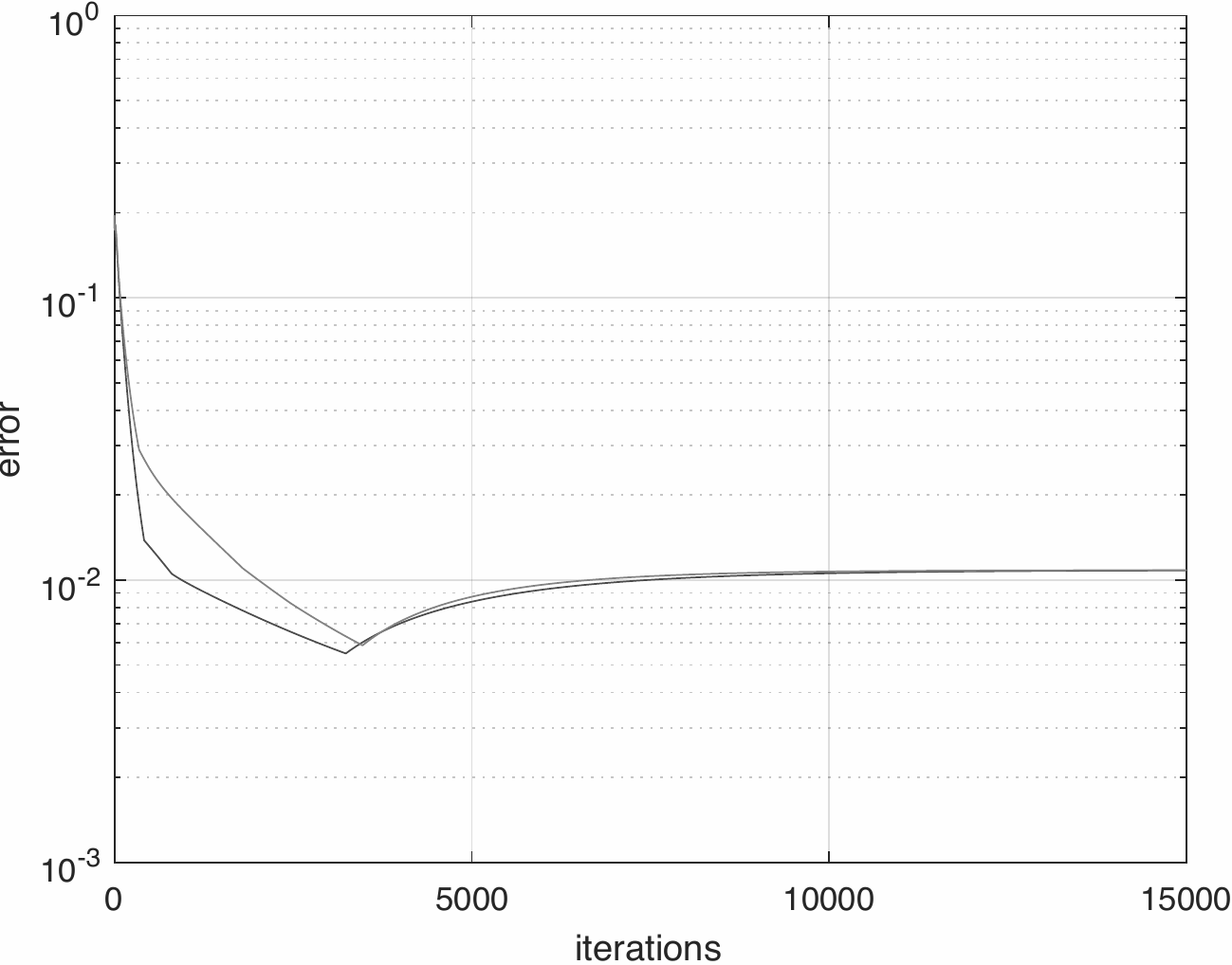}\hfill
\includegraphics[width=.3\textwidth]{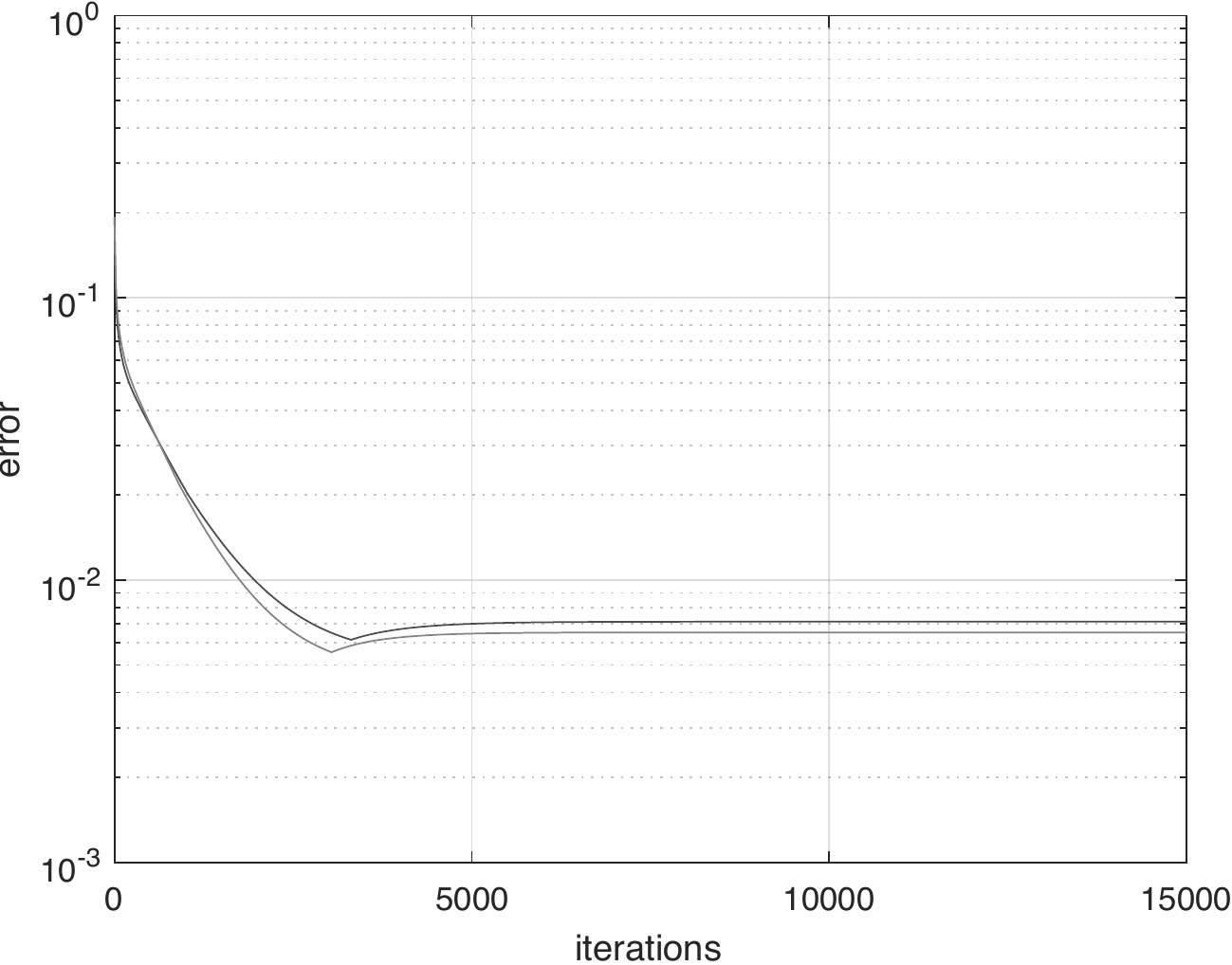}\hfill
\includegraphics[width=.3\textwidth]{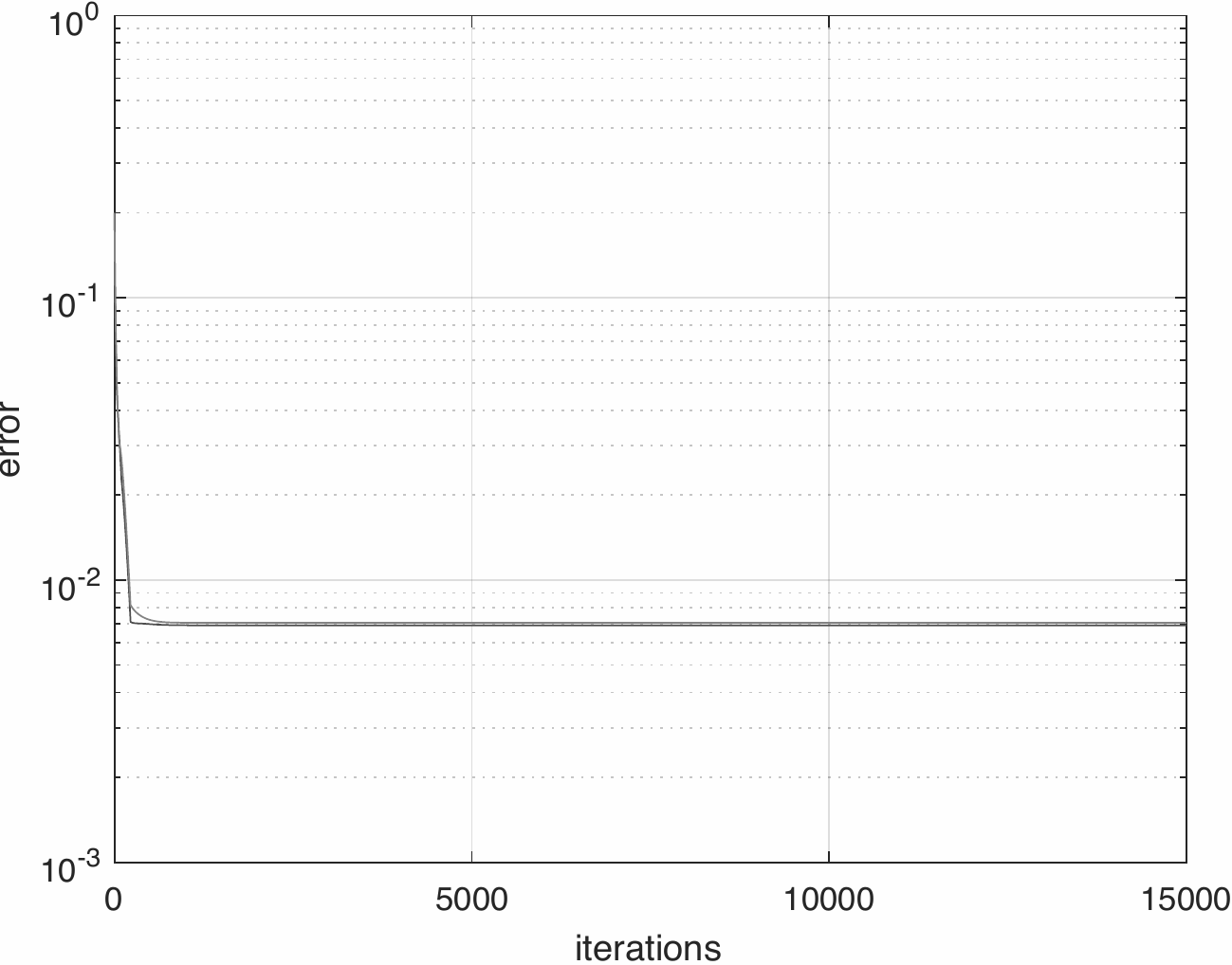}

\caption{\ao{{\bfseries Noisy case:} Everything is the same as in Figure \ref{fig:expvsnormal}, except that small independent noises of $0.01X$ where $X$ is uniform on $[-1/2,1/2]$ is added to every entry. This results in a perturbation matrix $\Delta$ with $||\Delta||_2 \sim 10^{-2}$, and all the final deviations from the correct solution $s$ are also on the order of $10^{-2}$.}}
\label{fig:expvsnormalpert}

\end{figure}

{\bf Time Complexity:} We also compare the time complexity of proposed algorithm against state-of-art algorithms. Our PGD algorithm requires fewer iterations in comparison to other iterative methods and each iteration scales linearly with $W$ and the maximum degree, $D_{max}$, of the worker-interaction graph which is bounded by $W$. Time complexity of different algorithms is summarized in Table~\ref{Tab:timecomplexity} 
\footnote{Opt-D\&S, KOS, and ER algorithms are omitted. They employ spectral factorization and have high time complexity.}.

\begin{table}
\centering
\caption{Time Complexity/Iteration of Different Methods.}
\label{Tab:timecomplexity}
\begin{tabular}{c|c|c|c}
\hline
Alg. & PGD & IWMV &M3W \\ 
\hline
Com. & $O(D_{max}W)$ & $O(TW)$ & $O(W^2T)$\\
\hline
\end{tabular}
\end{table}

\vspace{-0.05in}
\section{Conclusions}
We propose a new moment-matching approach
with weighted rank-one approximation and propose a gradient algorithm for worker skill estimation in Crowdsourcing.
In contrast to prior work, the weights are set up to correct for the spread 
of the measured worker-worker agreements accuracies which are typical in real-world
problems where who works on the same task with whom is out of control. 
Our results explicitly characterize identifiability and convergence rates
in terms of spectral graph theoretical quantities, revealing the importance of worker interaction graphs for skill estimation. The general problem studied here, is related to state estimation with intermittent and active sensor communications~\citep{dacsrv_interm06,Hanawal}, which we plan to explore in future work. 
\bibliography{refs.bib}

\newpage
\appendix
\onecolumn
\begin{center}
\Large Appendix
\end{center}
\newcommand{\x}{{\bf x}}
\newcommand{\p}{{\bf p}}
\newcommand{\y}{{\bf y}}
\newcommand{\bpi}{{\boldsymbol{\pi}}}

\def\N{\mathcal{N}}
\def\ox{\overline{x}}
\def\ovr{\overline{\rho}}
\def\oA{\overline{A}}
\def\A{\mathcal{A}}
\def\B{\mathcal{B}}
\def\oa{\overline{a}}
\def\oG{\overline{G}}
\def\sjn{\sum_{j=1}^n}
\def\E{\mathcal{E}}
\def\oE{\mathcal{\overline{E}}}
\def\oN{\mathcal{\overline{N}}}
\def\0{{\bf 0}}
\def\1{{\bf 1}}
\def\e{{\bf e}}
\def\s{{\bf s}}
\def\R{\mathbb{R}}
\def\C{\mathbb{C}}
\providecommand{\comjh[1]}{\com{JH}{#1}}
\providecommand{\com[2]}{\begin{tt}[#1: #2]\end{tt}}
\def\A{{\mathcal A}}
\def\jsr{{\rm jsr}}
\def\ao#1{\color{red}{#1}}
\def\aoc#1{{\color{red}#1}}

\def\bang#1{\smallbreak\noindent$\triangleright$\ {\it #1}\ }
\def\bangg#1{\smallbreak\noindent$\unrhd$\ \textit{#1}\ \ }

\def\A{\mathcal{A}}
\def\B{\mathcal{B}}
\def\red#1{{\color{red}#1}}
\def\ao{}
\def\aa{}

\def\argmin{\mathop{\rm argmin}}
\def\mf{\mathbf}
\def\mb{\mathbb}
\def\mc{\mathcal}
\def\beq{\begin{equation*}}
\def\eeq{\end{equation*}}
\def\bql{\begin{equation}}
\def\eql{\end{equation}}
\def\bqn{\begin{eqnarray*}}
\def\eqn{\end{eqnarray*}}
\def\bnl{\begin{eqnarray}}
\def\enl{\end{eqnarray}}
\def\bma{\begin{bmatrix}}
\def\ema{\end{bmatrix}}
\def\bmx{\begin{matrix}}
\def\emx{\end{matrix}}
\def\ben{\begin{enumerate}}
\def\een{\end{enumerate}}
\def\bit{\begin{itemize}}
\def\eit{\end{itemize}}
\def\bei{\begin{itemize}}
\def\eei{\end{itemize}}
\def\bet{\begin{tabular}}
\def\eet{\end{tabular}}
\def\und{\underline}
\def\unb{\underbrace}
\def\Log{\mbox{Log}}
\newcommand{\allcaps}[1]{\uppercase\expandafter{#1}}
\newcommand{\yn}{\color{blue}}
\def\cb #1{{\color{blue} #1}}
\def\ao #1{{\color{red} #1}}

\section{Proof of Lemma~\ref{lem:equiv}}
\label{sec:equiv}
Recall that $\theta$ is a finite index.
For each $(i,j)$ such that $N_{ij}>0$, we have
\begin{align*}
\frac{1}{N_{ij}}\sum_{t:(i,t),(j,t)\in A}(Y_{i,t}Y_{j,t}-x_ix_j)^2
&=\frac{1}{N_{ij}}\sum_{t:(i,t),(j,t)\in A}(1-2Y_{i,t}Y_{j,t}x_ix_j+x_i^2x_j^2)\\
&=1-2\tilde{C}_{ij}x_ix_j+x_i^2x_j^2\\
&=\tilde{C}_{ij}^2-2\tilde{C}_{ij}x_ix_j+x_i^2x_j^2+1-\tilde{C}_{ij}^2\\
&=(\tilde{C}_{ij}-x_ix_j)^2+1-\tilde{C}^2_{ij}.
\end{align*}
Therefore,
\begin{align*}
\frac{1}{2}\sum_{(i,t),(j,t)\in A}(Y_{i,t}Y_{j,t}-x_ix_j)^2
&=\frac{1}{2}\sum_{i,j\in[W]}N_{ij}(\tilde{C}_{ij}-x_ix_j)^2+\sum_{i,j\in[W]}N_{ij}(1-\tilde{C}_{ij}^2)
\end{align*}
Since $\sum_{i,j\in[W]}N_{ij}(1-\tilde{C}_{ij}^2)$ is a constant, Eq.\eqref{Eq:Objective} is equivalent to the optimization problem  $\argmin_{x\in[-1,+1]^W}L(x)$.


\section{Proof of Theorem~\ref{thm:learnability}}
The proof directly follows from Lemma~\ref{lem:asylearn} and Lemma~\ref{lem:asylearn1}.
We will next prove these Lemmas.

\bigskip

\noindent
{\it Proof of Lemma~\ref{lem:asylearn}:}
Take any two workers $i,j$ that are connected in $G$. Let $t\in \N$ be a task such that $(i,t),(j,t)\in A$.
By assumption, $$Y_{i,t} Y_{j,t} = g_t^2 Z_{i,t} Z_{j,t} = Z_{i,t} Z_{j,t}.$$ 
Let us define
\[ \bar C_{ij} \doteq \lim_{T\to\infty} \frac1T \sum_{(i,t),(j,t)\in A, t\le T}  Y_{i,t} Y_{j,t}.\] By the law of large numbers, 
\[ \bar C_{ij} = E[Z_{i,t} Z_{j,t}] = s_i s_j. \] For convenience, we define $ \bar C_{ij} = 0$ when $(i,j)\not\in E$.

Next without loss of generality assume that workers $1,2,\dots,2k+1$ form a cycle in $G$. 
Then, 
\begin{align*}
s_1  &= \bar C_{1,2k+1} s_{2k+1}^{-1} \\
       & = C_{1,2k+1} C_{2k+1,2k}^{-1} s_{2k} \\
	   & = C_{1,2k+1} C_{2k+1,2k}^{-1} C_{2k,2k-1} s_{2k-1}^{-1} \\
	   & \quad  \vdots \\
        &  =  C_{1,2k+1} C_{2k+1,2k}^{-1} C_{2k,2k-1} \dots C_{2,1} s_{1}^{-1} \,,
\end{align*}
which implies that
\begin{align*}
|s_1| = \sqrt{C_{1,2k+1} C_{2k+1,2k}^{-1} C_{2k,2k-1} \dots C_{2,1}} \,,
\end{align*}
assuming that $C_{2,3},C_{4,5},\dots,C_{2k,2k+1}\ne 0$. This gives a method to recover 
$|s_1|$. 

Now, since $G$ is connected, for any worker $i$ there exists a path from worker $1$ to worker $i$.
If this path was given by the vertices $1,2,\dots,\ell$ then
\begin{align*}
|s_\ell| =& |C_{\ell,\ell-1}| \,|s_{\ell-1}^{-1} | = |C_{\ell,\ell-1}| \,|C_{\ell-1,\ell-2}^{-1}|\, |s_{\ell-2} |\\
= &\dots = |C_{\ell,\ell-1}| \,|C_{\ell-1,\ell-2}^{-1}| \cdots |C_{2,1}^{(-1)^{\ell}}|\, |s_1|^{(-1)^{\ell+1}}\,.
\end{align*} which shows how $|s_l|$ may be recovered. We conclude that $|s|$ can be recovered. 

It remains to show that $\mathcal{P}(s)$ can be recovered.
Let $i,j\in [W]$ be two different workers. Then, if $\pi \subset E$ is any path in $G$ from $i$ to $j$,
we have $$\Pi_{(u,v)\in E} \sgn(C_{u,v}) = \Pi_{(u,v)\in E} \sgn(s_u) \sgn(s_v) = \sgn(s_i) \sgn(s_j).$$
We emphasize this holds for all paths connecting $i$ and $j$, and in particular $\Pi_{(u,v)\in E} \sgn(C_{u,v})$ is the same for any path connecting $i$ and $j$. 

Now if  $i$ and $j$ are such that for some path $\pi$ connecting them  $\Pi_{(u,v)\in E} \sgn(C_{u,v}) =+1$,
we assign $i,j$ to the same group; otherwise we assign them to different groups.  It is easy to see that this creates exactly two groups. The resulting 
``partition'' must match $\mathcal{P}(s)$. 

\bigskip

\noindent
{\it Proof of Lemma~\ref{lem:asylearn1}:}
Take $s,\alpha$ which are used in the definition of richness of $\Theta$.
We construct two other skill vectors $s'$ and $s''$ as follows:
We set $s'_1 = \alpha s_1$ and $s''_1 = s_1/\alpha$.
Now, if worker $i$ is at an even distance from worker $1$ on some path in $G$ then 
$s'_i = \alpha s_i$ and $s''_i = s_i/\alpha$, otherwise we set $s'_i = s_i/\alpha$ and $s''_i = \alpha s_i$. 
Note that all workers can be accessed from worker $1$ because $G$ is connected.
Note that if there are multiple paths from worker $1$ to some other worker then all of these have the same parity, or the graph had an odd cycle.
Now, both $s$ and $s'$ give rise to the same products, $s_i s_j$, along any edge $(i,j)\in E$. 
Since both are in $\Theta$ by assumption, the result is proven.




\bigskip

\noindent{\it Reverse Direction for Theorem~1:} 
We prove this by contraposition.
First, assume that (i) does not hold. We want to prove that learnability fails.
If (i) does not hold, we can take $s,s'\in [-1,1]^W$ different skill vectors such that $|s| = |s'|$ and $\mathcal{P}(s) = \mathcal{P}(s')$ and $s,s'\in S(\Theta)$. It follows that $s=-s'$.
Take any $g\in \pms^W$. Note that the instances $(s,A,g)$ and $(-s,A,-g)$ lead to the same joint distribution over the observed labels. Hence, no inference schema can tell these instances apart, thus any inference schema will suffer linear regret on one of these instances. 
Now, if (ii) does not hold, Lemma~\ref{lem:asylearn1} gives two skill vectors $s,s'$ which are different and $s\ne \pm s'$, which again give the same likelihood to any data. This again leads to that any inference schema will suffer a linear regret on one of these instances.

\section{Proof of Theorem~\ref{thm:noiselessPGD}}

Our first step is to argue that, with sufficiently small step-size, the PGD method remains bounded. To that end, we have the following proposition.

\begin{proposition}\label{lem:lyapunov} Let \[ V(x) = \max_{i=1, \ldots, W} \max \left( \frac{x_i}{s_i}, \frac{s_i}{x_i} \right) \] and suppose that $x^t$ is positive and $s$ is positive and that  the positive step-size $\gamma$ is small enough so that $$\eta ||N D_s^2||_{\infty} V(x^t)^2 \leq 1.$$ Then, it holds that $V(x^{t+1}) \leq V(x^t)$.  \label{boundlem} 
\end{proposition}

\begin{proof} Let us use the notation $a./b$ for the elementwise ratio of vectors $a$ and $b$, and define $r^t = x^t./s$. As a consequence,  $V(x^t) = \max_{i=1, \ldots, n} \max \left( r_i^t, 1/r_i^{t} \right)$. Now suppose $V(x^t) = Z$, which is a positive number due to the assumed positivity of $x^t$ and $s$; this  implies that \begin{equation} \label{yineq} Z^{-1} \leq r_i^t \leq Z \mbox{ for all } i = 1, \ldots, n. \end{equation} Let us fix some index $j$. The prove the lemma we just need to prove that Eq. (\ref{yineq}) holds for $r_j^{t+1}$. 

Suppose first that $r_j^t = \beta Z$ where $\beta \in (0,1]$. Then
\begin{eqnarray*} x_j^{t+1} & = & x_j^t - \eta \sum_{k=1}^n N_{jk} x_k^t (x_k^t x_j^t - s_k s_j) \\ 
 & = & x_j^t - \eta\sum_{k=1}^n N_{jk} s_k^2 s_j r_k^t (r_k^t r_j^t - 1) 
\end{eqnarray*}
and therefore
\begin{eqnarray} r_j^{t+1} & = & r_j^t - \eta \sum_{k=1}^n N_{jk} s_k^2 r_k^t (r_k^t r_j^t - 1) \label{firstline} 
\end{eqnarray} 
Now since $r_j^t = \beta Z$ and $r_k^t \geq Z^{-1}$ for all $k$, we have 
\begin{eqnarray}
r_{j}^{t+1} & \leq & \beta Z - \eta \sum_{k=1}^n N_{jk} s_k^2 r_k^t (Z^{-1} \beta Z - 1) \nonumber \\ 
& = & \beta Z + \eta (1-\beta) \sum_{k=1}^n N_{jk} s_k^2 r_k^t  \nonumber \\ 
& \leq & \beta Z + \eta (1-\beta) ||N D_s^2||_{\infty} Z  \nonumber \\
& \leq & \beta Z + (1-\beta) Z \nonumber \\ 
& = & Z, \nonumber
\end{eqnarray} where we used our step-size bound as well as the fact that $V(x^t) \geq 1$ due to the definition of $V(\cdot)$. This proves  the upper bound we seek. 

For the other direction, suppose $r_j^t = \mu Z^{-1}$ where now $\mu \in [1,\infty)$. From Eq. (\ref{firstline}), and using the fact that $r_k \leq Z$ for all $k$, we then have 
\begin{eqnarray*} r_j^{t+1} & \geq & \mu Z^{-1} - \eta \sum_{k=1}^n N_{jk} s_k^2 r_k (Z \mu Z^{-1}  - 1) \\
& = & \mu Z^{-1} - \eta (\mu - 1) \sum_{k=1}^n N_{jk} s_k^2 r_k \\ 
& \geq & \mu Z^{-1} - \eta (\mu - 1) ||N D_s^2||_{\infty} Z \\ 
& \geq & \mu Z^{-1} - \eta(\mu - 1) Z^{-1} ||N D_s^2||_{\infty} Z^2 \\
& \geq & \mu Z^{-1} -  (\mu - 1) Z^{-1} \\
& = & Z^{-1}
\end{eqnarray*}
\end{proof}

As a consequence of this proposition, we have the following bound on how big the iterates $x^t$ can get. 

\begin{corollary} \label{corbound} Suppose $s$ and $x^0$ belong to $[\kappa, K]^W$ where $0 < \kappa \leq K$. If the step-size $\eta$ of PGD algorithm satisfies $$0 \leq \eta \leq \frac{\kappa^2}{K^2 ||N D_s^2||_{\infty}},$$ then $V(x^t) \leq K/\kappa$ and $x^t \in [\kappa^2/K,K^2/\kappa]^W$ for all $t=1,2,\ldots$. 
\end{corollary} 

\begin{proof} 
Clearly, $V(x^0)^2 \leq (K/\kappa)^2$ given the fact that both $s$ and $x^0$ belong to $[\kappa,K]$. Then, $\eta$ satisfies the step-size condition of Proposition \ref{boundlem} at time $0$. Using Proposition \ref{boundlem}, we can conclude that $V(x^1) \leq V(x^0)$ which implies $V(x^1)^2 \leq V(x^0)^2\leq(K/\kappa)^2$ as well. By applying the same technique iteratively, we have $V(x^t)\leq\ldots\leq V(x^1) \leq V(x^0)\leq K/\kappa,\forall t$. Since $s \in [\kappa,K]^W$, this implies $x^t \in [\kappa^2/K,K^2/\kappa]^W$. 
\end{proof}

A consequence of the last corollary result is that the function 
\[ L_{\rm noiseless}(x) \doteq \frac{1}{2} \sum_{(i,j) \in E} N_{ij} (s_i s_j - x_i x_j)^2. \] can, for all practical purposes, be assumed to have a gradient which is Lipschitz. Of course, this is false over all of $\R^n$, but since the iterates of the PGD method stay within a compact set, the gradient of $L_{\rm noiseless}$ will be Lipschitz over the region of interest. In particular, we have the following estimate. 

\begin{proposition}\label{lem:firstorder}  For any $y,z \in [0, K^2/\kappa]^W$, we have that 
\[ ||\nabla L_{\rm noiseless}(y) - \nabla L_{\rm noiseless}(z)||_2 \leq \hat{L} ||y - z||_2, \] with 
\[ \hat{L} = 4W  ||N||_F \frac{K^4}{\kappa^2}. \] 
\end{proposition}

\begin{proof} We begin by establishing the following claim.  {\bf Claim:} Suppose $y,z$ are vectors belonging to the cube $y,z \in [0,A]^W$. Then $$||N \circ (y y^T - z z^T)||_F \leq 2 \sqrt{W} A  ||N||_F ||y-z||_2.$$

This proposition could be obtained by a simple algebraic manipulation as
\begin{eqnarray} ||N \circ (y y^T - z z^T ) ||_F
 & = & {\rm Tr}({\rm abs} (N) {\rm abs} (y y^T - z z^T)  )\label{firstl} \\
 & \leq & ||N||_F || y y^T - z z^T||_F  \label{secondl} \\ 
 & = & ||N||_F ||y y^T - y z^T + y z^T - z z^T||_F  \label{thirdl} \\ 
 & \leq & ||N||_F \left( ||y(y-z)^T||_F + ||(y-z)z^T||_F \right) \label{fourthl} \\
 & \leq & 2 ||N||_F \sqrt{W} A ||y-z||_2  \label{fifthl}
 \end{eqnarray}
Eq. (\ref{firstl}) follows via the inequality $||A \circ B||_F \leq {\rm Tr}({\rm abs}(A) {\rm abs}(B)^T)$. Eq. (\ref{secondl}) is obtained by the Cauchy-Schwarz inequality in the form of ${\rm Tr}(AB) \leq ||A||_F |||B||_F$.  Eq. (\ref{thirdl}) and Eq. (\ref{fourthl}) are self-explanatory. Eq. (\ref{fifthl}) follows because every entry of the vectors $y,z$ is at most $A$. This concludes the proof of the claim.

Let $P_x = N \circ (x x^T - s s^T)$. Then $\nabla L(x) = P_x x$ and therefore
\begin{eqnarray*} || \nabla L_{\rm noiseless}(y) - \nabla L_{\rm noiseless}(z)||_2 & = & ||P_y y - P_z z||_2 \\ 
& = & ||P_y y - P_y z + P_y z - P_z z ||_2 \\ 
& \leq & ||P_y||_2 ||y-z||_2 + ||P_y - P_z|| ||z||_2 \\
& \leq & 2 \sqrt{W} \frac{K^2}{\kappa}  ||N||_F ||y-z||_2 ||y-z||_2\\
& &+ 2 \sqrt{W} \frac{K^2}{\kappa} ||N||_F  ||y-z||_2 ||z||_2,
\end{eqnarray*} where the last step used that the $2$-norm of a matrix is upper bounded by its Frobenius norm as well as the above claim. Now using the bound $||z||_2 \leq (K^2/\kappa) \sqrt{W}$ and the same for $||y-z||_2$, we obtain the lemma. 

\end{proof}

We now use the standard fact that, for a function $f(x)$ with $L$-Lipschitz gradient, gradient descent with step-size $h<2/L$ generates a sequence which satisfies 
\[ f(x^{t+1}) = f(x^t) - h \left( 1 - \frac{L}{2} h \right) ||\nabla f(x^t)||_2^2. \] In particular, $\nabla f(x^t) \rightarrow 0$ under these conditions if $f(x^t)$ is bounded below.

Clearly, $L_{\rm noiseless}$ is bounded below by zero. Thus, to argue that gradient descent on $L_{\rm noiseless}$ results in $x^t \rightarrow s$, we just need to argue that the only point that satisfies $\nabla L_{\rm noiseless}(x)=0$ is $x=s$. 

We complete the proof of Theorem \ref{thm:noiselessPGD} by now proving that last statement. 

Observe that 
\[ [\nabla L_{\rm noiseless}(x)]_i = \sum_{j=1}^W N_{ij} (x_i x_j - s_i s_j) x_j, \] where  the interaction matrix $N$ which is nonnegative, irreducible, symmetric, and with zero diagonal. What we need to argue is that, given $s > 0$, there does not exist $x > 0, x \neq s$ such that for each $i=1, \ldots, W$, we have
\begin{equation} \label{maineq} \sum_{j=1}^W N_{ij} (x_i x_j - s_i s_j) x_j = 0. \end{equation}

We begin by adopting the following notation. For a vector $x$, $D_x$ will refer to the diagonal matrix with $x$ on the diagonal. For a matrix $A$, ${\rm diag}\left[ A\right]$ will refer to the \emph{diagonal of $A$ stacked as a vector} (note that this is an \textbf{unusual} notation). 
Also, let us refer to the set of matrices which are nonnegative, irreducible, aperiodic, symmetric and with  zero diagonal as {\em admissible}. 

Assume that $x$ satisfies Eq. \eqref{maineq}.
Then, we can multiply the $i$th equation of \eqref{maineq} by $x_i$.
Our first observation is that we may rewrite Eq. (\ref{maineq}) as
\begin{equation} \label{rewritten} {\rm diag}\left[ D_x N D_x (x x^T - s s^T) \right] = 0. \end{equation} 
It suffices to argue that we cannot positive $x$ and admissible $F$ such that
\[ {\rm diag}\left[ F (x x^T - s s^T ) \right]  = 0. \] 
Note that we were able to drop the $D_x$s from the equation because $N$ is admissible if and only if $D_x N D_x$ is.

We proceed as follows. Since 
\[ x_i x_j - s_i s_j = s_i \left( \frac{x_i}{s_i} \frac{x_j}{s_j} - 1 \right) s_j\,, \] defining $u_i = x_i/s_i$, we have that $u$ is positive and that 
\[ x x^T - s s^T = D_{s} (u u^T - 1 1^T) D_s\,. \] We must therefore argue that it is impossible to find $u>0, u \neq 1$ and admissible $F$ such that 
\[ {\rm diag} \left[  D_s F D_s (u u^T - 1 1^T) D_s \right] = 0.\] Since $s>0$ it will suffice to argue that we cannot find $u>0, u \neq 1$ and admissible $Z$ such that  
\begin{equation} \label{zequation} {\rm diag} \left[ Z (u u^T -1 1^T) \right] = 0.\end{equation} 

Without loss of generality, we can assume that  $u_1 \leq u_2 \leq \cdots \leq u_W$; we can always relabel indices to make this hold. 

Now there are three possibilities:
\begin{enumerate}\item $u_1 u_W > 1$.
\item $u_1 u_W = 1$. 
\item $u_1 u_W < 1$. 
\end{enumerate}

We argue that in each case we cannot find a suitable $u$ that satisfies Eq. (\ref{zequation}). Indeed, let us  consider the first possibility. In that case the last column of $u u^T - 1 1^T$, with entries $u_i u_W-1$, is strictly positive, and therefore, considering that $[Z(uu^T - 1 1^T)]_{WW}=0$, we obtain that the last row of $Z$ must be zero --  contradicting irreducibility.  Similarly, in case $3$, the first column of $u u^T - 1 1^T$, with entries $u_1 u_i-1$, is negative, and, considering that $[Z(uu^T - 1 1^T)]_{11}=0$, we see that the first row of $Z$ must be zero, which can not hold true.

It remains to consider case $2$.  We may assume that $u_1  < u_W$ (ruling out the possibility that a $u$ proportional to the all-ones vector satisfies Eq. (\ref{zequation}) is trivial). We break up $\{1,\ldots, W\}$ into three blocks. The first block is all the indices $j$ such that $u_j = u_1$. The third block is all the indices $j$ such that that $u_j = u_W$. All the other indices go into block $2$. Note that block 2 may be empty, for example if every entry of $u$ is equal to $u_1$ or $u_W$.

The advantage of partitioning this way is that 
the matrix $u u^T - 1 1^T$ has the following sign structure:

\[ u u^T - 1 1^T = \left( \begin{array}{c|c|c} 
- & - & 0 \\
\hline
- & * & + \\
\hline
0 & + & + 
\end{array} \right)\] where $-$ represents a strictly negative submatrix, $+$ represents a strictly positive submatrix, while $*$ represents a submatrix that can have elements of any sign.  The strict negativity comes from the fact that $u_1 < u_W$. 

 Partitioning $Z$ in a compatible manner, we have that 
\[ {\rm diag} \left[   \left( \begin{array}{c|c|c} 
Z_{11} & Z_{12} & Z_{13} \\
\hline
Z_{21} & Z_{22} & Z_{23} \\
\hline
Z_{31} & Z_{32} & Z_{33} 
\end{array} \right) \left( \begin{array}{c|c|c} 
- & - & 0 \\
\hline
- & * & + \\
\hline
0 & + & + 
\end{array} \right) \right] = 0.\]

Considering the $(1,1)$ diagonal block of the above product,
noting that $Z\ge 0$,
 we obtain $Z_{11} = Z_{12}=0$; and considering the $(3,3)$ diagonal block of the above product we obtain  $Z_{32} =  Z_{33}=0$. 
 By symmetry, also $Z_{21} = 0$ and $Z_{33}=0$.

From here we can easily derive a contradiction. Indeed,  if the middle block is nonempty, the matrix is reducible; and if the second block is empty, it is periodic.

\section{Proof of Theorem~\ref{thm:pert}}

\aor{Let us adopt the convention that when $a=(a_1, \ldots, a_n)$ is a vector, we will understand $e^a$ to apply to it elementwise, i.e., $e^a = (e^{a_1}, \ldots, e^{a_n})$.  We will find it convenient to do our analysis in terms of the variables $z(t)$ defined through the relation $x(t) = e^{z(t)}$. In terms of these variables, we can rewrite Eq. (\ref{expgrad}) as 
\begin{equation} \label{eq:firstzeq} e^{z(t+1)} = P_{{\cal C}} \left[ e^{z(t) - \alpha \nabla_t} \right]. \end{equation}} \aor{Observe that projection of a vector $a$ onto the cube ${\cal C}$ simply thresholds each $a_i$ between $\kappa$ and $K$. Inspecting Eq. (\ref{eq:firstzeq}), we therefore see that we can move the projection inside the exponentiation if we instead project onto the cube 
$\Omega = [\ln \kappa, \ln K]^W$:  
\[ e^{z(t+1)} = e^{P_{\Omega} \left[ z(t) - \alpha \nabla_t \right]}, \]  or
\begin{equation} \label{eq:secondzeq} z(t+1) = P_{\Omega} \left[ z(t) - \alpha \nabla_t  \right]. 
\end{equation}} \aor{The trick that makes the proof possible is that we can construct a function so that Eq. (\ref{eq:secondzeq}) becomes a projected gradient descent iteration. To that end, we define
\[ g_{\Delta}(z) = \frac{1}{2} \sum_{i,j=1}^W N_{ij} e^{z_i + z_j} - \sum_{i=1}^W  z_i \sum_{j=1}^n N_{ij} (s_i s_j + \Delta_{ij}). \] The key observation is that the gradient of this function is
\begin{eqnarray*}  [\nabla g_{\Delta}(z(t))]_i & = & \sum_{j=1}^W N_{ij} e^{z_i(t) + z_j(t)} - \sum_{j=1}^n N_{ij} \left(s_i s_j + \Delta_{ij} \right) \\ 
& = & \sum_{j=1}^W N_{ij} (e^{z_i(t) + z_j(t)} - s_i s_j - \Delta_{ij}) \\ 
& = & [\nabla_t]_i,
\end{eqnarray*} where the last step used the definition of $z(t)$, i.e., $e^{z_k(t)} = x_k(t)$ for all $k=1, \ldots, W$. }

\aor{Thus we have that
\begin{equation}\label{eq:pgdz} z(t+1) = P_{\Omega} \left[ z(t) - \alpha \nabla g_{\Delta}(z(t)) \right]. 
\end{equation} Because we now have a projected gradient descent iteration on a convex function (it is, of course, immediate that $g_{\Delta}(z)$ is convex), it should now be  clear that an analysis of this iteration in terms of $z(t)$ id possible, provided by we can
upper bound the condition number of the function $g_{\Delta}(z)$ over the region $z \in [\ln \kappa, \ln K]^W$.  To analyze this condition number, we argue as follows. }

\aor{First, because $\nabla^2 g_{\Delta}(z) =  L_s((N_{ij}/2)e^{z_i+ z_j})$, where $L_s(w_{ij})$ refers to the signless Laplacian with weights $w_{ij}$, i.e., 
\[ L_s(w_{ij}) = \sum_{i,j=1}^W w_{ij} (e_i + e_j) (e_i + e_j)^T. \] It follows that $\lambda_{\rm min}(L_s(w_{ij}))$ is a monotonic function of the weights $\{ w_{ij}\}$; this then implies that  
$g_{\Delta}(z)$ is a $\mu$-strongly convex over $\Omega$, where $\mu = \kappa^2 N_{\rm min} \lambda_{\rm min}(L_s) $, where $L_s$ is the signless Laplacian of the unweighted graph corresponding to the interaction matrix $N_{ij}$ and $N_{\rm min}$ is the smallest positive $N_{ij}$. We remind the reader that it was shown in \citet{desai1994characterization} that $\lambda_{\rm min}(L_s) \geq 1/W^3$. } 

\aor{Second, we need to argue that $g_{\Delta}(z)$ has gradient that is $L$-Lipschitz, along with an estimate for $L$. To this end, we reprise the argument we used in the proof of Theorem \ref{thm:noiselessPGD} and argue that for any $a,b \in [\ln \kappa, \ln K]^W$ we have: 
\begin{eqnarray*} || \nabla g_{\Delta}(a) - \nabla g_{\Delta}(b)||_2 & = & ||{\rm diag}(e^a) N e^a - {\rm diag}(e^b) N e^b||_2 \\
& \leq & ||{\rm diag}(e^a) N (e^a - e^b)||_2  + ||({\rm diag}(e^a) - {\rm diag}(e^b)) N e^b||_2 \\   
& \leq & K ||N||_2 K ||a-b||_2 + K\max_i |a_i - b_i|  ||N||_2 K \sqrt{W}\\ 
& \leq & 2 ||N||_2 K^2 \sqrt{W} ||a-b||_{2},  
\end{eqnarray*}  where we used that for vectors $a,b \in [ \ln \kappa, \ln K]^W$ we have that 
\begin{eqnarray*} ||{\rm diag}(e^a)||_2 &  \leq  & K \\ 
||e^b||_2 & \leq & K \sqrt{W} \\ 
||e^a - e^b||_2 & \leq & K ||a-b||_2 \\ 
 ||{\rm diag}(e^a) - {\rm diag}(e^b)||_2 & \leq &  K \max_{i=1, \ldots, W} |a_i - b_i |
\end{eqnarray*}
 In conclusion, we may take $L=2 \sqrt{W} ||N||_2 K^2$.} 
 
 \aor{Having obtained bounds on $L, \mu$, we next analyze the performance of Eq. (\ref{eq:pgdz}). Defining, 
\[ z_{\Delta}^* := \arg \min_{z \in \Omega} g_{\Delta}(z), \] we have that $z(t)$ converges to $z^*$ with the choice $\alpha = 1/L$; as an immediate consequence, we have that $x(t) \rightarrow x_{\Delta}^*$ where 
\[ x_{\Delta}^* = e^{z_{\Delta}^*}. \] This proves the first assertion of the theorem, that $x(t)$ converges.} 

\aor{To prove the second and third part, we first need to establish the following technical lemma.} 

\aor{\begin{lemma} \label{lem:int} If $s$ lies in the interior of $[\kappa, K]^W$ and $\max_{i,j} |\Delta_{ij}|$ is small enough, then $z_{\Delta}^*$ lies in the interior of $\Omega = [\ln \kappa, \ln K]^W$. 
\end{lemma}} 

\begin{proof} \aor{We first argue that the unique minimizer of $g_0(z)$ (i.e., when $\Delta=0$) is $z^* = \ln s$. Indeed, observe that $\nabla g_0(\ln s) =0$. Moreover, we have already discussed that, over the region $[\ln a, \ln b]^W$, the Hessian of $g_0$ has eigenvalues lower bounded by $\mu = a^2 N_{\rm min} \lambda_{\rm min}(L_s)$. Since the assumption that the graph corresponding to $N$ is connected and non-bipartite renders implies the signless Laplacian $L_s$ is positive definite (again, see \cite{desai1994characterization}), we obtain that $g_0(z)$ is strictly convex. Thus $z^*=\ln s$ is the unique minimizer of $g_0$ over $\Omega$.} 

\aor{It follows that the same property holds for small enough $\Delta$ ``by continuity.'' More formally, the argument is as follows. As a consequence of the function that $\ln s$ is the unique minimizer of $g_0(z)$ over $\Omega$, we have that there is a ball ${\cal B}$ of positive radius around $\ln s$ such that the function $g_0(z)$ is strictly smaller on ${\cal B}$ than it is on any point of the boundary of the cube $\Omega$. It follows that, for small enough $\Delta$, the function $g_{\Delta}$ will also be strictly smaller on ${\cal B}$ than on any point on the boundary of $\Omega$. This implies that the minimium of $g_{\Delta}$ occurs in the interior of $\Omega$.}
\end{proof} 

\aor{Having established this lemma, we now turn to an analysis of the convergence times. We have that that standard results for projected gradient descent on strongly convex functions, with the choice of of step-size $\alpha = 1/L$, we have that (see Theorem 2.4 of \cite{hazan2016introduction})
\begin{equation} \label{eq:bdecay} g_{\Delta}(z(t)) - g_{\Delta} (z^*) \leq e^{-t\mu/(4L)} \left( g_{\Delta} (z(0)) - g_{\Delta}(z^*) \right), \end{equation} where $\mu$ is the strong convexity coefficient. Our next step is to translate this into a convergence rate for $x(t)$.} 

\aor{First, using the mean value theorem, we have that for any two scalars $u,v$m
\begin{equation} \label{eq:meanvalue}
\min \{e^u, e^v\} |u-v| \leq  |e^u - e^v| \leq \max \{ e^u, e^v\} |u-v|, 
\end{equation}
and we use this in the next sequence of equations:
\begin{eqnarray*} 
||x(t) - x_{\Delta}^*||_2^2 & \leq  K^2||z(t) - z_{\Delta}^*||_2^2 & \mbox{By  Eq.}(\eqref{eq:meanvalue}). \\ 
& \leq K^2 \frac{2}{\mu} \left( g_{\Delta}(z(t)) - g_{\Delta}(z^*) \right) & \\
& \leq  K^2 \frac{2}{\mu} e^{-t\mu/(4L)} \left( g_{\Delta}(z(0)) - g_{\Delta}(z^*) \right) &  \mbox{By  Eq.} (\eqref{eq:bdecay})
\end{eqnarray*}
Now any function $h(y)$ convex over a convex region $\mathcal{R}$ with $L$-Lipschitz gradient over the same region satisfies (see proof of Lemma 1.2.3 in \cite{Nest04})
\[ h(y_1) \leq h(y_2) + \nabla h(y_2)^T (y_1 - y_2) + \frac{L}{2} ||y_1 - y_2||_2^2, \] for any $y_1, y_2 \in \mathcal{R}$. If $y_2$ is further chosen to be a point satisfying $\nabla h(y_2) = 0$, then  we have \[ h(y_1) - h(y_2) \leq \frac{L}{2} ||y_1 - y_2||_2^2 \] We next apply this to the function $g_{\Delta}$, which is convex with $L$-Lipschitz gradient over the region $\Omega$. By Lemma \ref{lem:int}, the minimizer $z_{\Delta}^*$ lies in the interior of $\Omega$, and consequently we have $\nabla g_{\Delta}(z^*)=0$. Therefore,}

\aor{ \begin{eqnarray*} 
||x(t) - x_{\Delta}^*||_2^2 & \leq  K^2 \frac{L}{\mu}  e^{-t\mu/(4L)} ||z(0) - z_{\Delta}^*||_2^2 &   \\ & \leq  \frac{K^2}{\kappa^2} \frac{L}{\mu} e^{-t\mu/(4L)} ||x(0) - x_{\Delta}^*||_2^2  & \mbox{By  Eq.} (\ref{eq:meanvalue})
\end{eqnarray*}} \aor{Since $||x(0) - x^*||_2^2 \leq W K^2$ because $|x_i(0)| \leq K$, we have that it takes $$\frac{4L}{\mu} \ln \frac{W K^4 L}{\epsilon \mu \kappa^2}$$ iterations until $||x(t) - x_{\Delta}^*||_2^2 \leq \epsilon$.  Since the quantity $L/\mu$ scales polynomially in the number of workers $W$, this proves the second assertion of the theorem, namely that convergence to any $\epsilon$-neighborhood of the limit $x_{\Delta}^*$ occurs in polynomial time. This proves the second assertion of the theorem.}

\aor{We now turn to the last assertion of the theorem, i.e., the bound on $||x_{\Delta}^*-s||_2$. For this part, start with the equation
\[ \nabla g_{\Delta} (z_{\Delta}^*) = 0, \] which we argued above will hold for small enough $\Delta$, and observe that its consequence is that 
\begin{equation} \label{firstpert} || \nabla g_0(z_{\Delta}^*) ||_2 = \left| \left|[\sum_{j=1}^W N_{ij} \Delta_{ij}]_i \right| \right|_2 \leq \sqrt{W} ||N||_{\infty} \max_{i,j} |\Delta_{ij}|. \end{equation}} 
\aor{Our final step is to argue that this implies $z_{\Delta}^*$ is close to $\ln s$. Indeed, let us define  $\phi(t) = \nabla g_0(\ln s + t (z_{\Delta}^* - \ln s))$. We thus have that 
\begin{eqnarray*}
\nabla g_0(z_{\Delta}^*) & = & \phi(1) - \phi(0) \\ 
& = & \int_0^1 \phi'(u) du \\
& = & \int_0^1 \nabla^2 g_0(u)  (z_{\Delta}^* - \ln s) ~ du
\end{eqnarray*} so, multiplying both sides by $z_{\Delta}^*- - \ln s$, we obtain  
\[ (z_{\Delta}^* - \ln s)^T \nabla g_0(z_{\Delta}^*) \geq \mu ||z_{\Delta} - \ln s||_2^2, \] where, as before, $\mu = \kappa^2 N_{\rm min} \lambda_{\rm min}(L_s) $ is a lower bound on the smallest eigenvalue of $\nabla^2 g_0(u)$ when $u \in [\ln \kappa, \ln K]^W$. Now using Cauchy-Schwarz, this implies
\[ ||\nabla g_0(z_{\Delta}^*)||_2 \geq \mu ||z_{\Delta}^* - \ln s||. \] Putting this together with Eq. (\ref{firstpert}), we obtain 
\[ ||z_{\Delta}^* - \ln s|| \leq \frac{\sqrt{W} ||N||_{\infty}}{\mu} \max_{i,j} |\Delta_{ij}|. \] 
Finally, 
\begin{eqnarray*} 
||x_{\Delta}^* - s||_2 & = & ||e^{z_{\Delta}^*} - e^{\ln s}||_2 \\ 
& \leq & K ||z_{\Delta}^* - \ln s||_2 \\ 
& \leq & K  \frac{\sqrt{W} ||N||_{\infty}}{\mu} \max_{i,j} |\Delta_{ij}|.
\end{eqnarray*} This concludes the proof of the theorem.}

\end{document}